%% file: FastReilableCrowdsourcing.tex
\documentclass[10pt,journal,compsoc]{IEEEtran}

\usepackage{times}
\usepackage{color}
\usepackage[cmex10]{amsmath}
\usepackage{algorithmic}
\usepackage{siunitx}
\usepackage{dsfont}
\usepackage[linesnumbered,ruled,vlined]{algorithm2e}
\DecMargin{1em} 
\usepackage{url}
\usepackage{epsfig,tabularx,amssymb,subfigure,multirow,graphicx, hyphenat}
\usepackage{balance}
\usepackage{amsthm} 
\usepackage{enumitem}
\usepackage{hyperref}

\long\def\comment#1{}

\newtheorem{theorem}{\bf Theorem}[section]
\newtheorem{lemma}{\bf Lemma}[section]
\newtheorem{corollary}{\bf Corollary}[section]

\theoremstyle{remark}

\theoremstyle{definition}
\newtheorem{definition}{\bf Definition}

\SetCommentSty{mycommfont}

\begin{document}

\title{FROG: A Fast and Reliable Crowdsourcing Framework (Technical Report)}

\author{
	Peng~Cheng, ~Xiang~Lian, ~Xun~Jian, ~Lei~Chen,~\IEEEmembership{Member,~IEEE}
	\IEEEcompsocitemizethanks{
		\IEEEcompsocthanksitem P. Cheng is with
		the Department of Computer Science and Engineering, Hong Kong
		University of Science and Technology, Kowloon, Hong Kong, China,
		Email: pchengaa@cse.ust.hk.\protect\\
		\IEEEcompsocthanksitem X. Lian is with the Department of Computer
		Science, Kent State University, Ohio, USA, Email: xlian@kent.edu.\protect\\
		\IEEEcompsocthanksitem X. Jian is with
		the Department of Computer Science and Engineering, Hong Kong
		University of Science and Technology, Kowloon, Hong Kong, China,
		Email: xjian@cse.ust.hk.\protect\\
		\IEEEcompsocthanksitem L. Chen is with the Department of Computer
		Science and Engineering, Hong Kong University of Science and
		Technology, Kowloon, Hong Kong, China, Email:
		leichen@cse.ust.hk.} 
}

\markboth{Technical Report}
{Shell \MakeLowercase{\textit{Cheng and Chen}}: Transaction on
	Knowledge and Data Engineering}

\IEEEcompsoctitleabstractindextext{
	
\begin{abstract}

	For decades, the crowdsourcing has gained much attention from both
	academia and industry, which outsources a number of tasks to human
	workers. Typically, existing crowdsourcing platforms include
	CrowdFlower, \textit{Amazon Mechanical Turk} (AMT), and so on, in
	which workers can autonomously select tasks to do. However, due to
	the unreliability of workers or the difficulties of tasks, workers
	may sometimes finish doing tasks either with incorrect/incomplete
	answers or with significant time delays. 
	Existing studies considered improving the task accuracy 
	through voting or learning methods, they usually did not fully take into account
	reducing the latency of the task completion. This is especially
	critical, when a task requester posts a group of tasks (e.g.,
	sentiment analysis), and one can only obtain answers of \textit{all}
	tasks after the last task is accomplished. As a consequence, the
	time delay of even one task in this group could delay the next step
	of the task requester's work from minutes to days, which is quite
	undesirable for the task requester.

	Inspired by the importance of the task accuracy and latency, in this
	paper, we will propose a novel crowdsourcing framework, namely
	\textit{\underline{F}ast and \underline{R}eliable
		cr\underline{O}wdsourcin\underline{G} framework} (FROG), which
	intelligently assigns tasks to workers, such that the latencies of
	tasks are reduced and the expected accuracies of tasks are met.
	Specifically, our FROG framework consists of two important
	components, \textit{task scheduler} and \textit{notification}
	modules. For the task scheduler module, we formalize a \textit{FROG
		task scheduling} (FROG-TS) problem, in which the server actively
	assigns workers to tasks to achieve high task reliability and low task latency. We
	prove that the FROG-TS problem is NP-hard. Thus, we design two
	heuristic approaches, request-based and batch-based scheduling. For
	the notification module, we define an \textit{efficient worker
		notifying} (EWN) problem, which only sends task invitations to those
	workers with high probabilities of accepting the tasks. To tackle
	the EWN problem, we propose a \textit{smooth kernel density
		estimation} approach to estimate the probability that a worker
	accepts the task invitation. Through extensive experiments, we
	demonstrate the effectiveness and efficiency of our proposed FROG
	platform on both real and synthetic data sets.
\end{abstract}

	
}

\maketitle

\section{Introduction}
\label{sec:introduction}

Nowadays, the crowdsourcing has become a very useful and practical
tool to process data in many real-world applications, such as the
sentiment analysis \cite{mohammad2013crowdsourcing}, image labeling
\cite{welinder2010online}, and entity resolution
\cite{wang2012crowder}. Specifically, in these applications, we may
encounter many tasks (e.g., identifying whether two photos have the
same person in them), which may look very simple to humans, but not
that trivial for the computer (i.e., being accurately computed by
algorithms). Therefore, the crowdsourcing platform is used to
outsource these so-called \textit{human intelligent tasks} (HITs) to
human workers, which has attracted much attention from both academia
\cite{fan2015icrowd, liu2012cdas, li2016crowdsourced} and industry
\cite{ipeirotis2010analyzing}.

Existing crowdsourcing systems (e.g., CrowdFlower \cite{crowdflower}
or \textit{Amazon Mechanical Turk} (AMT) \cite{amt}) usually wait
for autonomous workers to select tasks. As a result, some difficult
tasks may be ignored (due to lacking of the domain knowledge) and
left with no workers for a long period of time (i.e., with high
latency). What is worse, some high-latency (unreliable) workers may
hold tasks, but do not accomplish them (or finish them carelessly),
which would significantly delay the time (or reduce the quality) of
completing tasks. Therefore, it is rather challenging to guarantee
high accuracy and low latency of tasks in the crowdsourcing system,
in the presence of unreliable and high-latency workers.

Consider an example of auto emergency response on interest public places, which monitors the incidents happened at important places (e.g., crossroads). In such applications, results with low latencies and high accuracies are desired. However, due to the limitation of current computer vision and AI technology, computers cannot do it well without help from humans. For example, people can know one car may cause accidents only when they know road signs in pictures and the traffic regulations, what computers cannot do. Applications may embed ``human power'' into the system as a module, which assigns monitoring pictures to crowdsourcing workers and aggregates the answers (e.g., ``Normal'' or ``Accident'') from workers in almost real-time. Thus, the latency of the crowdsourcing module will affect the overall application performance. 

\noindent \textbf{Example 1 (Accuracy and Latency Problems in the
	Crowdsourcing System)} {\it The application above automatically selects and posts 5 pictures as 5 emergency reorganization tasks $t_1 \sim t_5$ at different timestamps,
	respectively, on a crowdsourcing platform. Assume that 3 workers, $w_1 \sim w_3$, from the crowdsourcing system autonomously accept some or all of the 5 tasks, $t_i$ (for
	$1\leq i\leq 5$, posted by
	the emergency response system.

	Table \ref{example_answer_table} shows the answers and time delays of tasks conducted by workers, $w_j$ ($1\leq
	j\leq 3)$, where the last column provides the correctness
	(``$\surd$'' or ``$\times$'') of the emergency reorganization answers against the
	ground truth. Due to the unreliability of workers and the
	difficulties of tasks, workers cannot always do the tasks correctly.
	That is, workers may be more confident to do specific categories of
	tasks (e.g., biology, cars, electronic devices, and/or sports), but
	not others. For example, in Table \ref{example_answer_table}, worker
	$w_2$ tags all pictures (tasks), $t_1\sim t_5$, with 3 wrong labels. Thus, in this case, it is rather challenging to guarantee
	the accuracy/quality of emergency reorganization (task) answers, in the presence of
	such unreliable workers in the crowdsourcing system.

	Furthermore, from Table \ref{example_answer_table},
	all the 5 tasks are completed by workers within 20 seconds, except
	for task $t_2$ which takes worker $w_2$ 5 minutes to finish
	(because of the difficulty of task $t_2$). Such a long latency is highly
	undesirable for the emergency response application, who needs to proceed with the
	emergency reorganization results for the next step.
	Therefore, with the existence of high latency workers in the crowdsourcing system, it
	is also important, yet challenging, to achieve low latency of the
	task completion. 
	
}

\begin{table}[t!] 
	\begin{center}
		\caption{ Answers of Tasks from Workers.} \label{example_answer_table}
		
		\begin{tabular}{c|c|c|c|c}
			{\bf Worker} & {\bf Task} & \textbf{Answer} & \textbf{Time Latency} & \textbf{Correctness} \\
			\hline \hline
			$w_1$ & $t_1$& Normal& 8 s & $\times$\\\hline
			$w_1$ & $t_2$& Accident& 9 s & $\times$\\\hline
			$w_1$ & $t_3$& Accident& 12 s & $\surd$\\\hline\hline
			
			$w_2$ & $t_1$ & Accident&  15 s & \underline{$\times$}\\ \hline
			$w_2$ & $t_2$ & Normal& {\bf 5 min} & \underline{$\surd$}\\ \hline
			$w_2$ & $t_3$ & Normal& 10 s & \underline{$\times$}\\ \hline
			$w_2$ & $t_4$ & Accident& 9 s & \underline{$\surd$}\\ \hline
			$w_2$ & $t_5$ & Accident& 14 s &  \underline{$\times$}\\ \hline\hline
			
			$w_3$ & $t_4$& Accident& 8 s &  $\surd$\\ \hline
			$w_3$ & $t_5$ & Normal& 11 s &  $\surd$\\ \hline\hline
			
		\end{tabular}
	\end{center}
\end{table}

Specifically, the FROG framework contains two important components,
\textit{task scheduler} and \textit{notification} modules. In the
task scheduler module, our FROG framework actively schedules tasks for workers,
considering both accuracy and latency. In particular, we formalize a
novel \textit{FROG task scheduling} (FROG-TS) problem, which finds ``good'' worker-and-task assignments that minimize
the maximal latencies for all tasks and maximize the accuracies
(quality) of task results. We prove that the FROG-TS problem is
NP-hard, by reducing it from the \textit{multiprocessor scheduling
	problem} \cite{gary1979computers}. As a result, FROG-TS is not
tractable. Alternatively, we design two heuristic approaches,
request-based and batch-based scheduling, to efficiently tackle the
FROG-TS problem.

Note that, existing studies on reducing the latency are usually
designed for specific tasks (e.g., filtering or resolving entities)
\cite{parameswaran2014optimal, verroios2015entity} by increasing
prices over time to encourage workers to accept tasks
\cite{gao2014finish}, which cannot be directly used for
general-purpose tasks under the budget constraint (i.e., the
settings in our FROG framework). Some other studies
\cite{haas2015clamshell, fan2015icrowd} removed low-accuracy or
high-latency workers, which may lead to idleness of workers and low
throughput of the system. In contrast, our task scheduler module
takes into account both factors, accuracy and latency, and can design a worker-and-task assignment strategy with high
accuracy, low latency,  and high throughput.

In existing crowdsourcing systems, workers can freely join or leave
the system. However, in the case that the system lacks of active
workers, there is no way to invite more offline workers to perform
online tasks. To address this issue, the notification module in our
FROG framework is designed to notify those offline workers via
invitation messages (e.g., by mobile phones). However, in order to
avoid sending spam messages, we propose an \textit{efficient worker
	notifying} (EWN) problem, which only sends task invitations to those
workers with high probabilities of accepting the tasks. To tackle
the EWN problem, we present a novel \textit{smooth kernel density
	estimation} approach to efficiently compute the probability that a
worker accepts the task invitation.

To summarize, in this paper, we have made the following
contributions.

\begin{itemize}
	\item We propose a new FROG framework for crowdsourcing, which
	consists of two important task scheduler and notification modules in
	Section \ref{sec:frame}.
	
	\item We formalize and tackle a novel worker-and-task scheduling problem in crowdsourcing, namely FROG-TS,
	which assigns tasks to suitable workers, with high reliability and
	low latency in Section \ref{sec:routing_module}.
	
	\item We propose a \textit{smooth kernel density model} to estimate the probabilities that
	workers can accept task invitations for the EWN problem in the
	notification module in Section \ref{sec:notification}.
	
	\item We conduct extensive experiments to verify the effectiveness and efficiency of our proposed FROG framework
	on both real and synthetic data sets in Section \ref{sec:exper}.
\end{itemize}

Section \ref{sec:related} reviews previous studies on the
crowdsourcing. Section \ref{sec:conclusion} concludes this paper.

\section{Problem Definition}
\label{sec:frame}

\subsection{The FROG Framework}

Figure \ref{fig:framework} illustrates our \textit{fast and reliable
	crowdsourcing} (FROG) framework, which consists of \textit{worker
	profile manager}, \textit{public worker pool}, \textit{notification
	module}, \textit{task scheduler module}, and \textit{quality
	controller}.

Specifically, in our FROG framework, the worker profile manager
keeps track of statistics for each worker in the system, including
the response time (or in other words, the latency) and the accuracy
of doing tasks in each category. These statistics are dynamically
maintained, and can be used to guide the task scheduling process (in
the task scheduler module).

Moreover, the public worker pool contains the information of online
workers who are currently available for doing tasks. Different from
the existing work \cite{haas2015clamshell} with exclusive retainer
pool, we use a shared public retainer pool, which shares workers for
different tasks. It can improve the global efficiency of the
platform, and benefit workers with more rewards by assigning with
multiple tasks (rather than one exclusive task for the exclusive
pool).

When the number of online workers in the public worker pool is
small, the notification module will send messages to offline workers
(e.g., via mobile devices), and invite them to join the platform.
Since offline workers do not want to receive too many (spam)
messages, in this paper, we will propose a novel \textit{smooth
	kernel density model} to estimate the probabilities that offline
workers will accept the invitations, especially when the number of
historical samples is small. This way, we will only send task
invitations to those offline workers with high probabilities of
accepting the tasks.

Most importantly, in FROG framework, the task scheduler module assigns tasks to suitable workers with the goals of reducing
the latency and enhancing the accuracy for tasks. In this module, we
formalize a novel \textit{FROG task scheduling} (FROG-TS) problem,
which finds good worker-and-task assignments to minimize the
maximal task latencies and to maximize the accuracies (quality) of
task results. Due to the NP-hardness of this FROG-TS
problem (proved in Section \ref{sec:np-hard_FRC}), we design
two approximate approaches, request-based and batch-based
scheduling approaches.

Finally, the quality controller is in charge of the quality
management during the entire process of the FROG framework. In
particular, before workers are assigned to do tasks, we require that
each worker need to register/subscribe one's expertise categories of
tasks. To verify the correctness of subscriptions, the quality
controller provides workers with some qualification tests, which
include several sample tasks (with ground truth already known).
Then, the system will later assign them with tasks in their
qualified categories. Furthermore, after workers submit their
answers of tasks to the system, the quality controller will check
whether each task has received enough answers. If the answer is yes,
it will aggregate answers of each task (e.g., via voting
methods), and then return final results.

\begin{figure}[t!]\centering
	\scalebox{0.3}[0.3]{\includegraphics{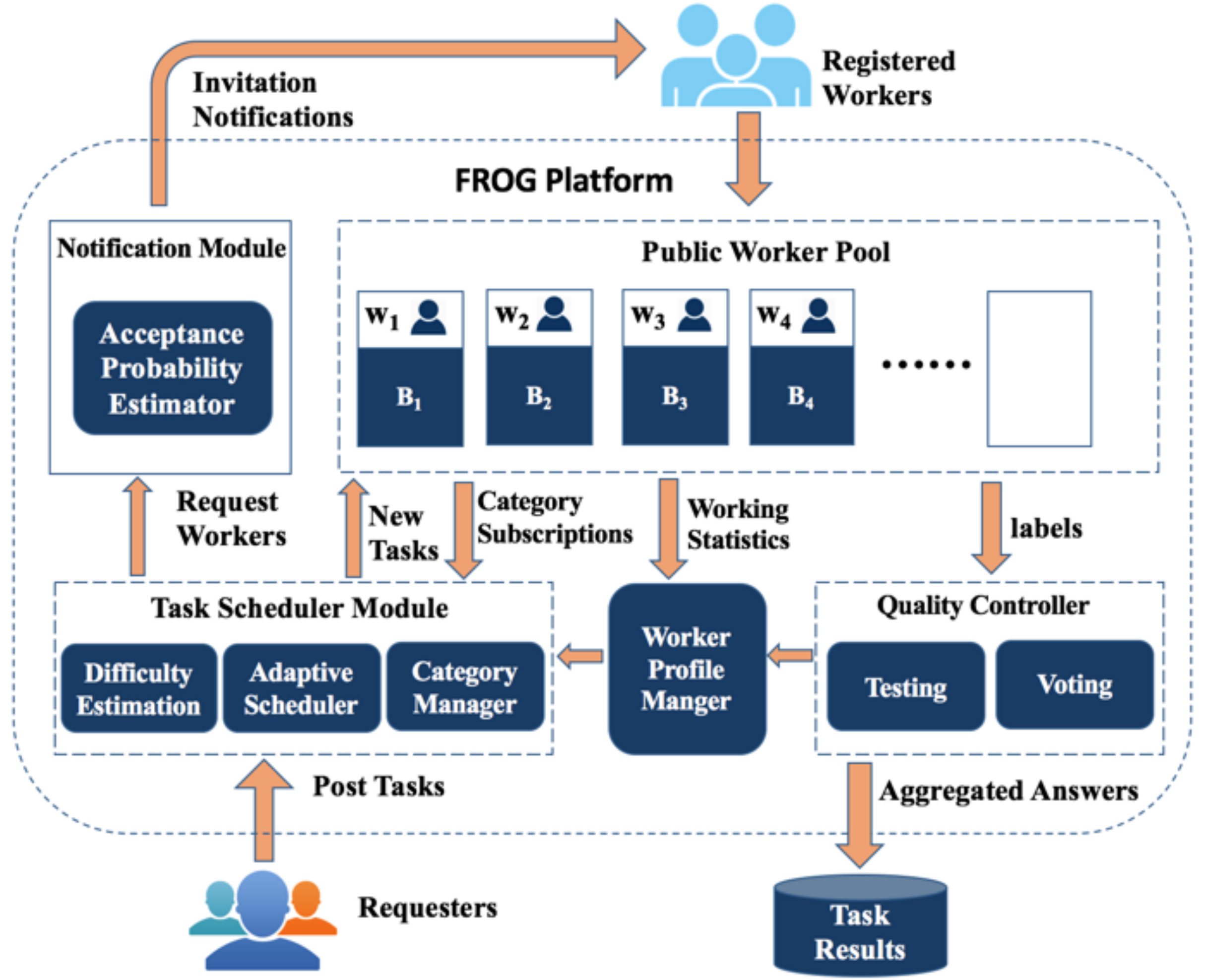}}
	\caption{ An Illustration of the FROG Framework.}
	\label{fig:framework}
\end{figure}

In this paper, we will focus on general functions of two important
modules, task scheduler and notification modules, in the FROG
framework (as depicted in Figure \ref{fig:framework}), which will be
formally defined in the next two subsections. Note that the notification module is only an optional component of our framework. If we remove the notification module from our FROG framework, our FROG framework can continue to work on crowdsourcing tasks (e.g., worker pool management, task scheduling and quality controlling). For example, in \cite{fan2015icrowd} and \cite{haas2015clamshell}, the authors use the ExternalQuestion mechanism of AMT \cite{amt} to manage microtasks on their own Web server and take full control of microtask assignments. On the other hand, if sending messages through mobile phones or notification mechanisms are allowed in our framework, our system can do better by inviting reliable workers.
In addition, we implement our framework and use WeChat \cite{wechat} as its client to send tasks to workers and receive answers from workers. Workers will get paid by WeChat red packets after they contribute to the tasks. Other message Apps such as Whatsapp \cite{whatsapp} and Skype\cite{skype} can also be used as clients.

\subsection{The Task Scheduler Module}

The task scheduler module focuses on finding a good worker-and-task
assignment strategy with low latency (i.e., minimizing the maximum
latency of tasks) and high reliability (i.e., satisfying the
required quality levels of tasks).

\noindent {\bf Tasks and Workers.} We first give the definitions for
tasks and workers in the FROG framework. Specifically, since our
framework is designed for general crowdsourcing platforms, we
predefine a set, $C$, of $r$ categories for tasks, that is,
$C=\{c_1, c_2, ..., c_r\}$, where each task belongs to one category
$c_l \in C$ ($1\leq l\leq r$). Here, each category can be the
subject of tasks, such as cars, food, aerospace, or politics.

\begin{definition}
	$($Tasks$)$ Let $T=\{t_1, t_2, ..., t_m\}$ be a set of $m$ tasks in
	the crowdsourcing platform, where each task $t_i$ ($1\leq i\leq m$)
	belongs to a task category, denoted by $t_i.c\in C$, and arrives at
	the system at the starting time $s_i$. Moreover, each task $t_i$ is
	associated with a user-specified quality threshold $q_i$, which is
	the expected probability that the final result for task $t_i$ is
	correct.
	\label{definition:task}
\end{definition}

Assume that task $t_i$ is accomplished at the completion time $f_i$.
Then, the latency, $l_i$, of task $t_i$ can be given by: $l_i = f_i
- s_i$, where $s_i$ is the starting time (defined in Definition
\ref{definition:task}). Intuitively, the smaller the latency $l_i$
is, the better the performance of the crowdsourcing platform is.

\begin{definition} $($Workers$)$ Let $W=\{w_1, w_2, ..., w_n\}$ be a
	set of $n$ workers. For tasks in category $c_l$, each worker $w_j$ ($1\leq j\leq
	n$) is associated with an accuracy, $\alpha_{jl}$, that $w_j$ do tasks in
	category $c_l$, and a response time, $r_{jl}$.
	\label{definition:worker}
\end{definition}

As given in Definition \ref{definition:worker}, the category
accuracy $\alpha_{jl}$ is the probability that worker $w_j$ can
correctly accomplish tasks in category $c_l$. Here, the response
time $r_{jl}$ measures the period length from the timestamp that
worker $w_j$ receives a task $t_i$ (in category $c_l$) to the time
point that he/she submits $t_i$'s answer to the server.

In the literature, in order to tackle the intrinsic error rate (unreliability) of
workers, there are some existing voting methods for result
aggregations in the crowdsourcing system, such as the majority
voting \cite{cao2012whom}, weighted majority voting
\cite{fan2015icrowd}, half voting\cite{mo2013optimizing}, and
Bayesian voting \cite{liu2012cdas}. For the ease of presentation, in this paper, we use the
majority voting for the result aggregation, which has been well
accepted in many crowdsourcing studies \cite{cao2012whom}. Assuming
that the count of answering task $t_i$ is odd, if the majority
workers (not less than $\lceil\frac{k}{2} \rceil$ workers) vote for
a same answer (e.g., Yes), we take this answer as the final
result of task $t_i$. 
Denote $W_i$ as the set of $k$ workers that do task $t_i$, and $c_l$ as
the category that task $t_i$ belongs to. Then, we have the expected
accuracy of task $t_i$ as follows:
\begin{equation}
\Pr(W_i, c_l) = \sum_{x = \lceil\frac{k}{2} \rceil}^{k}\sum_{W_{i,x}}\Big(\prod_{w_j \in W_{i,x}}\alpha_{jl}\prod_{w_j \in W_i - W_{i,x}}(1-\alpha_{jl})\Big),
\label{eq:task_accuracy}
\end{equation}

\noindent where $W_{i,x}$ is a subset of $W_i$ with $x$ elements.

Specifically, the expected task accuracy, $\Pr(W_i, c_l)$,
calculated with Eq. (\ref{eq:task_accuracy}) is the probability that
more than half of the workers in $W_i$ can answer $t_i$ correctly.
In the case of voting with multiple choices (other than 2 choices, like YES/NO), please refer to \textbf{Appendix A} of supplementary materials
for the equations of the expected accuracy of task $t_i$ with
majority voting or other voting methods.
Table \ref{table0} summarizes the commonly used symbols.
\begin{table}
	\begin{center}
		\caption{Symbols and descriptions.} \label{table0}
		\begin{tabular}{l|l}
			{\bf Symbol} & {\bf \qquad \qquad \qquad Description} \\
			\hline \hline
			$C$ & a set of task categories $c_l$\\
			$T$   & a set of tasks $t_i$\\
			$W$   & a set of workers $w_j$\\
			$t_i.c$ & the category of task $t_i$\\
			$q_i$   & a specific quality value of task $t_i$ \\
			$s_i$   & the start time of task $t_i$\\
			$f_i$   &  the finish time of task $t_i$\\
			$\alpha_{jl}$ & the category accuracy of worker $w_j$ on tasks in category $c_l$\\
			$r_{jl}$ & the response time of worker $w_j$ on tasks in category $c_l$\\
			\hline \hline
		\end{tabular}
	\end{center}
\end{table}

\noindent {\bf The FROG Task Scheduling Problem.} In the task
scheduler module, one important problem is on how to route tasks to
workers in the retainer pool with the guaranteed low latency and
high accuracy. Next, we will formally define the problem of
\textit{FROG Task Scheduling} (FROG-TS) below.

\begin{definition} $($FROG Task Scheduling Problem$)$ Given a set $T$ of $m$ crowdsourcing tasks, and $n$ workers in $W$,
	the problem of \textit{FROG task scheduling} (FROG-TS) is to assign
	workers $w_j\in W$ to tasks $t_i \in T$, such that:
	\begin{enumerate}[leftmargin=*]
		\item the accuracy $\Pr(W_i, c_l)$ (given in Eq. (\ref{eq:task_accuracy})) of task $t_i$ is not lower than the required accuracy threshold $q_i$, 
		
		\item the maximum latency $\max(l_i)$ of tasks in $T$ is minimized,
	\end{enumerate}
	\noindent where $l_i = f_i - s_i$ is the latency of task $t_i$, that is, the duration from the time $s_i$ task $t_i$ is posted in the system
	to the time, $d_i$, task $t_i$ is completed. 
	
	\label{definition:crowdsourcing}
\end{definition}

We will later prove that the FROG-TS problem is NP-hard (in Section
\ref{sec:np-hard_FRC}), and propose two effective approaches,
request-based and batch-based scheduling, to solve this problem in
Section \ref{sec:adaptive_framework}.

\subsection{The Notification Module}

The notification module is in charge of sending notifications to
those offline workers with high probabilities of being available and
accepting the invitations (when the retainer pool needs more
workers). In general, some workers may join the retainer pool
autonomously, but it cannot guarantee that the retainer pool will be
fulfilled quickly. Thus, the notification module will invite more
offline workers to improve the fulfilling speed of the retainer
pool.

Specifically, in our FROG framework, the server side maintains a
public worker pool to support all tasks from the requesters. When
\textit{autonomous workers} join the system with a low rate, the
system needs to invite more workers to fulfill the worker pool, and
guarantees high task processing speed of the platform. One
straightforward method is to notify all the offline workers.
However, this broadcast method may disturb workers, when they are
busy with other jobs (i.e., the probabilities that they accept
invitations may be low). For example, assume that the system has
10,000 registered workers, and only 100 workers may potentially
accept the invitations. With the broadcast method, all 10,000
workers will receive the notification message, which is inefficient
and may potentially damage the user experience. A better strategy is
to send notifications only to those workers who are very likely to
join the worker pool. Moreover, we want to invite workers with high-accuracy and low-latency.
Therefore, we formalize this problem as the \textit{efficient worker
	notifying} (EWN) problem.

\begin{definition} $($Efficient Worker Notifying Problem$)$ Given a timestamp $ts$,
	a set $W$ of $n$ offline workers, the historical online records
	$E_j=\{e_1, e_2,..., e_n\}$ of each worker $w_j$, and the number,
	$u$, of workers that we need to recruit for the public worker pool,
	the problem of \textit{efficient worker notifying} (EWN) is to
	select a subset of workers in $W$ with high accuracies and low latencies to send invitation messages, such
	that:

	\begin{enumerate}
		\item the expected number, $E(P_{ts}(W))$, of workers who accept the invitations is greater than $u$, and
		
		\item the number of workers in $W$, to whom we send notifications, is minimized,
	\end{enumerate}
	
	\noindent where $P_{ts}(.)$ is the probability of workers to accept invitations and log in the worker pool at timestamp $ts$. 
	\label{definition:notification}
\end{definition}

In Definition \ref{definition:notification}, it is not trivial to
estimate the probability, $P_{ts}(W)$, that a worker prefers to log in
the platform at a given timestamp, especially when we are lacking of
his/her historical records. However, if we notify too many workers,
it will disturb them, and in the worst case drive them away from our
platform forever. To solve the EWN problem, we propose an effective
model to efficiently do the estimation in Section
\ref{sec:notification}, with which we select workers with high
acceptance probabilities, $P_{ts}(.)$, to send invitation messages
such that the worker pool can be fulfilled quickly.
Moreover, since we want to invite workers with high acceptance
probabilities, low response times, and high accuracies, we define
the worker dominance below to select good worker candidates.

\begin{definition} (Worker Dominance) Given two worker candidates $w_x$ and $w_y$,
	we say worker $w_x$ dominates worker $w_y$, if it holds that: (1)
	$P_{ts}(w_x) > P_{ts}(w_y)$, (2) $\alpha_x \geq \alpha_y$, and (3)
	$r_x \leq r_y$, where $P_{ts}(w_j)$ is the probability that worker
	$w_j$ is available, and $\alpha_x$ and $r_x$ are the average
	accuracy and response time of worker $w_x$ on his/her subscribed
	categories, respectively. 
	\label{definition:worker_dominance}
\end{definition}

Then, our notification module will invite those offline workers, $w_j$, with high ranking scores (i.e., defined as the number of workers dominated by worker $w_j$ \cite{yiu2007efficient}). We will discuss the details of the ranking later in Section 4.

\section{The Task Scheduler Module}
\label{sec:routing_module}

The task scheduler module actively routes tasks to workers, such
that tasks can be completed with small latency and the quality
requirement of each task is satisfied. In order to improve the
throughput of the FROG platform, in this section, we will estimate
the difficulties of tasks and response times (and accuracies as
well) of workers, based on records of recent answering. In
particular, we will first present effective approaches to estimate
worker and task profiles, and then tackle the \textit{fast and
	reliable crowdsourcing task scheduling} (FROG-TS) problem, by
designing two efficient heuristic-based approaches (due to its
NP-hardness).

\subsection{Worker Profile Estimation}

We first present the methods to estimate the
category accuracy and the response time of a worker, which can be
used for finding good worker-and-task assignments in the FROG-TS
problem.

\noindent \textbf{The Estimation of the Category
	Accuracy.} In the FROG framework, before each worker $w_j$ joins the
system, he/she needs to subscribe some task categories, $c_l$,
he/she would like to contribute to. Then, worker $w_j$ will complete
a set of qualification testing tasks $T_c=\{t_1, t_2, ..., t_m\}$ of category $c_l$,
by returning his/her answers, $A_j=\{a_{j1}, a_{j2},$ $...,
a_{jm}\}$, respectively. Here, the system has the ground truth of
the testing tasks in $T_c$, denoted as $G=\{g_1, g_2, ..., g_m\}$.

Note that, at the beginning, we do not know the difficulties of the
qualification testing tasks. Therefore, we initially treat all
testing tasks with equal difficulty (i.e., 1). Next, we estimate the
category accuracy, $\bar{\alpha}_{jl}$, of worker $w_j$ on category
$c_l$ as follows:

\begin{equation}
\bar{\alpha}_{jl} = \frac{\sum_{i=1}^{|T_c|} \mathds{1} (a_{ji}=g_i)}{|T_c|}
\label{eq:naive_accuracy}
\end{equation}

\noindent where $\mathds{1} (v)$ is an indicator function (i.e., if
$v$ is $true$, we have $\mathds{1} (v) = 1$; otherwise, $\mathds{1} (v)
= 0$), and $|T_c|$ is the number of qualification testing tasks. Note that, $t_i.c = c_l, \forall t_i\in T_c$.
Intuitively, Eq.~(\ref{eq:naive_accuracy}) calculates the percentage
of the correctly answered tasks (i.e., $a_{ji}=g_i$) by worker $w_j$
(among all testing tasks).

In practice, the difficulties of testing tasks can be different.
Intuitively, if more workers provide wrong answers for a task, then
this task is more difficult; similarly, if a high-accuracy worker
fails to answer a task, then this task is more likely to be
difficult.

Based on the intuitions above, we can estimate the difficulty of a
testing task as follows. Assume that we have a set, $W_c$, of
workers $w_j$ (with the current category accuracies  $\alpha_{jl}$)
who have passed the qualification test. Then, we give the definition
of the difficulty $\beta_i$ of a testing task $t_i$ below: 
\begin{equation}
\beta_i = \frac{\sum_{j=1}^{|W_c|} \big(\mathds{1} (a_{ji} \neq
	g_i)\cdot \bar{\alpha}_{jl}\big)}{\sum_{j=1}^{|W_c|}
	\bar{\alpha}_{jl}} \label{eq:difficulty}
\end{equation}

\noindent where $\mathds{1} (v)$ is an indicator function, and
$|W_c|$ is the number of workers who passed the qualification test.

In Eq.~(\ref{eq:difficulty}), the numerator (i.e.,
$\sum_{j=1}^{|W_c|} \big(\mathds{1} (a_{ji} \neq g_i) \cdot
\bar{\alpha}_{jl}\big)$) computes the weighted count of wrong
answers by workers in $W_c$ for a testing task $t_i$. The
denominator (i.e., $\sum_{j=1}^{|W_c|} \bar{\alpha}_{jl}$) is used
to normalize the weighted count, such that the difficulty $\beta_i$
of task $t_i$ is within the interval $[0, 1]$. The higher $\beta_i$
is, the more difficult $t_i$ is.
In turn, we can treat the difficulty $\beta_i$ of task $t_i$ as a
weight factor, and rewrite the category accuracy,
$\bar{\alpha}_{jl}$, of worker $w_j$ on category $c_l$ in
Eq.~(\ref{eq:naive_accuracy}) as:
\begin{equation}
\bar{\alpha}_{jl} = \frac{\sum_{i=1}^{|T_c|} \big(\mathds{1}
	(a_{ji}=g_i) \cdot \beta_i\big)}{\sum_{i=1}^{|T_c|} \beta_i}.
\end{equation}

\vspace{0.5ex}\noindent \textbf{The Update of the Category
	Accuracy.} 
After worker $w_j$ passes the qualification test of task category
$c_l$, he/she will be assigned with tasks in that category.
However, the category accuracy of a worker may vary over time. For
example, on one hand, the worker may achieve more and more accurate
results, as he/she is more experienced in doing specific tasks. On
the other hand, the worker may become less accurate, since he/she is
tired after a long working day. To keep tracking the varying accuracies of workers, we may update their accuracy based on their performance on their latest $k$ tasks $T_r$ in category $c_l$.

Assume that the aggregated results for $k$ latest real tasks in $T_r$ is $\{g'_1,
g'_2, ...,$ $g'_k\}$, and answers provided by $w_j$ are $\{a'_{j1},
a'_{j2}, ..., a'_{jk}\}$, respectively. Then, we update the
category accuracy $\alpha_{jl}$ of worker $w_j$ on category $c_l$ as
follows:
\begin{equation}
\alpha_{jl} = \theta_j \cdot \bar{\alpha}_{jl} + (1-\theta_j) \cdot \frac{\sum_{i=1}^{k}
	\mathds{1} (a'_{ji}=g'_i)}{k},
\label{eq:category_accuracy}
\end{equation}

\noindent where $\theta_j = \frac{|W_c|}{|W_c| + k}$ is a balance parameter to combine the performance of each worker in testing tasks and real tasks.
We can use the aggregated results of the real tasks in the latest 10$\sim$20 minutes to update the workers' category accuracy with Eq. (\ref{eq:category_accuracy}).

\noindent \textbf{The Estimation of the Category
	Response Time.} In reality, since different workers may have
different abilities, skills, and speeds, their response times could
be different, where the response time is defined as the length of
the period from the timestamp that the task is posted to the time
point that the worker submits the answer of the task to the server.

Furthermore, the response time of each worker may change temporally
(i.e., with temporal correlations). To estimate the response time,
we utilize the latest $\eta$ response records of worker $w_j$ for
answering tasks in category $c_l$, and apply the
\textit{least-squares method} \cite{leon1980linear} to predict the
response time, $r_{jl}$, of worker $w_j$ in a future timestamp. The input of the least-squares method is the $\eta$ latest (timestamp, response time) pairs. The
least-squares method can minimize the summation of the squared
residuals, where the residuals are the differences between the
recent $\eta$ historical values and the fitted values provided by
the model. We use the fitted line to estimate the category
response time in a future timestamp.

The value of $\eta$ may affect the sensitiveness and stability of the estimation of the category response time. Small $\eta$ may lead the estimation sensitive about the response times of workers, however, the estimated value may vary a lot. Large $\eta$ causes the estimation stable, but insensitive. In practice, we can set $\eta$ as the number of responses of worker $w_j$ to the tasks in category $c_l$ in recent 10 $\sim$ 20 minutes.

\subsection{Task Profile Estimation}

In this subsection, we discuss the task difficulty, which may affect
the latency of accomplishing tasks.

\noindent \textbf{The Task Difficulty.} Some tasks in the
crowdsourcing system are in fact more difficult than others. In AMT \cite{amt},
autonomous workers pick tasks by themselves. As a consequence,
difficult tasks will be left without workers to conduct. In
contrast, in our FROG platform, the server can designedly
assign/push difficult tasks to reliable and low-latency workers to
achieve the task quality and reduce the time delays.

For a given task $t_i$ in category $c_l$ with $R$ possible answer
choices ($R=2$, in the case of YES/NO tasks), assume that $|W_i|$
workers are assigned with this task. Since some workers may skip the task
(without completing the task), we denote $\gamma_i$ as the number of
workers who skipped task $t_i$, and $\Omega_i$ as the set of received
answers, where $|\Omega_i|+ \gamma_i = |W_i|$. Then, we can estimate the
difficulty $d_i$ of task $t_i$ as follows:
\begin{equation}
d_i = \frac{\gamma_i}{|W_i|} + \frac{|\Omega_i|}{|W_i|} \cdot
\frac{Entropy(t_i, \Omega_i)}{MaxEntropy(R)} + \epsilon,
\label{eq:task_difficulty}
\end{equation}
\noindent where $\epsilon$ is a small constant representing the base difficulty of tasks. Here, in Eq.~(\ref{eq:task_difficulty}), we have:
\begin{equation}
Entropy(t_i, \Omega_i) = \sum_{r=1, W_{i,r} \neq \emptyset}^{R} - \frac{\sum_{w_j \in W_{i,r}} \alpha_{jl} }{\sum_{w_j \in
		W_i} \alpha_{jl}} \log
\left(\frac{\sum_{w_j \in W_{i,r}} \alpha_{jl}}{\sum_{w_j \in
		W_i} \alpha_{jl}}\right), \label{eq:task_entropy}
\end{equation}
\begin{equation}
MaxEntropy(R) = R\cdot \left( -\frac{1}{R}\cdot
\log\left(\frac{1}{R}\right)\right) = \log(R),
\label{eq:upper_entropy}
\end{equation}

\noindent where $W_{i,r}$ is the set of workers who select the $r$-th
possible choice of task $t_i$ and $|\Omega_i|$ is the number of
received answers. Note that, when at the beginning no worker answers task $t_i$, we assume its entropy $Entropy(t_i, \Omega_i) = 0$.

\noindent {\bf Discussions on the Task Difficulty.} The task
difficulty $d_i$ in Eq.~(\ref{eq:task_difficulty}) is estimated
based on the performance of workers $W_i$ on doing task $t_i$. Those workers who
skipped the task treat tasks as being the most difficult (i.e., with
difficulty equal to 1), whereas for those who completed the task, we use
the normalized entropy (or the diversity) of their answers to
measure the task difficulty.

Specifically, the first term (i.e., $\frac{\gamma_i}{|W_i|}$) in
Eq.~(\ref{eq:task_difficulty}) indicates the percentage of workers
who skipped task $t_i$. Intuitively, when a task is skipped by more
percentage of workers, it is more difficult.

The second term in Eq.~(\ref{eq:task_difficulty}) is to measure the
task difficulty based on answers from those $\frac{|\Omega_i|}{|W_i|}$
percent of workers (who did the task). Our observation is as
follows. When the answers of workers are spread more evenly (i.e.,
more diversely), it indicates that it is harder to obtain a final
convincing answer of the task with high confidence. In this paper,
to measure the diversity of answers from workers, we use the
entropy \cite{shannon2001mathematical}, $Entropy(t_i, \Omega_i)$ (as given in
Eq.~(\ref{eq:task_entropy})), of answers, with respect to the
accuracies of workers. Intuitively, when a task is difficult to
complete, workers will get confused, and eventually select diverse
answers, which leads to high entropy value. Therefore, larger
entropy implies higher task difficulty. Moreover, we also normalize this
entropy in Eq.~(\ref{eq:task_difficulty}), that is, dividing it by
the maximum possible entropy value, $MaxEntropy(R)$ (as given in
Eq.~(\ref{eq:upper_entropy})).

For the subjective tasks, such as image labeling, translation and knowledge acquisition, we may use other data mining or machine learning methods \cite{liu2006image, dascalu2014analyzing} to estimate their difficulties. For example, the number of words and the level of words of a sentence can be used as features to estimate its difficulty for workers to translate.

\subsection{Hardness of the FROG-TS Problem}
\label{sec:np-hard_FRC}

We prove that the FROG-TS problem is NP-hard, by
reducing it from the \textit{multiprocessor scheduling problem}
(MSP) \cite{gary1979computers}.

\begin{theorem} (Hardness of the FROG-TS Problem) The problem of FROG Task Scheduling (FROG-TS) is NP-hard. \label{lemma:expensive_worker}
\end{theorem}

\begin{proof}
	We prove the lemma by a reduction from the multiprocessor scheduling problem (MSP).
	A multiprocessor scheduling problem can be described as follows: Given a set $J$ of $m$
	jobs where job $j_i$ has length $l_i$ and a number of $n$ processors, the multiprocessor
	scheduling problem is to schedule all jobs in $J$ to $n$ processors without overlapping
	such that the time of finishing all the jobs is minimized.
	
	For a given multiprocessor scheduling problem, we can transform it to an instance of
	FROG-TS problem as follows: we give a set $T$ of $m$ tasks and each task $t_i$
	belongs to a different category $c_i$ and the specified accuracy is lower than the lowest
	category accuracy of all the workers, which means each task just needs to be answered by one worker.
	For $n$ workers, all the workers have the same response time $r_i = l_i$ for the
	tasks in category $c_i$, which leads to the processing time of any task $t_i$ is
	always $r_{i}$ no matter which worker it is assigned to.
	
	As each task just needs to be assigned to one worker, this FROG-TS problem instance is to minimize the maximum completion time of task $t_i$ in $T$, which is identical to minimize the time of finishing
	all the jobs in the given multiprocessor scheduling problem. With this mapping
	it is easy to show that the multiprocessor scheduling problem instance can be
	solved if and only if the transformed FROG-TS problem can be solved.
	
	This way, we reduce MSP to the FROG-TS problem. Since MSP is known
	to be NP-hard \cite{gary1979computers}, FROG-TS is also
	NP-hard, which completes our proof.
\end{proof}

The FROG-TS problem focuses on completing multiple tasks that
satisfy the required quality thresholds, which requires that each
task is answered by multiple workers. Thus, we cannot directly use
the existing approximation algorithms for the MSP problem (or its
variants) to solve the FROG-TS problem.
Due to the NP-hardness of our FROG-TS problem, in the next
subsection, we will introduce an adaptive task routing approach with
two worker-and-task scheduling algorithms, \textit{request-based}
and \textit{batch-based scheduling} approaches to efficiently
retrieve the FROG-TS answers.

\subsection{Adaptive Scheduling Approaches}
\label{sec:adaptive_framework}

In this subsection, we first estimate the delay probability of each task. The higher the delay probability is, the more
likely the task will be delayed. Then we propose two adaptive scheduling strategies, request-based scheduling and batch-based scheduling, to iteratively assign workers to the task with the highest delay probability such that the maximum processing time of tasks is minimized.

\subsubsection{The Delay Probability}

As mentioned in the second criterion of the FROG-TS problem (i.e.,
in Definition \ref{definition:crowdsourcing}, we want to minimize
the maximum latency of tasks in $T$. In order to achieve this goal,
we will first calculate the delay probability, $L(t_i)$, of task
$t_i$ in $T$, and then assign workers to those tasks with high delay
probabilities first, such that the maximum latency of tasks can be
greedily minimized.

We denote the current timestamp as $\bar{\epsilon}$. Let $\epsilon_i$ ($= \min \{\bar{\epsilon} - s_i, l_i\}$) be the time lapse of task $t_i$ and $\epsilon_{max}$ ($=\max_{t_i \in T} \epsilon_i$) be the current maximum time lapse and $\bar{r}_l$ be the average response time of task $t_i$ in category $c_l$. Then, $\lceil\frac{\epsilon_{max} -\epsilon_i}{\bar{r}_l}\rceil$ is the number of more rounds for task $t_i$ to enlarge the maximum time lapse. We denote $P(t_i, \omega_k)$ as the probability of that the task $t_i$  will not meet its accuracy requirement in a single round $\omega_k$. As a difficult task will result in evenly distributed answers according to the definition of $d_i$ in Eq. (\ref{eq:task_difficulty}), a task $t_i$ having a larger $d_i$ will have a higher probability $P(t_i, \omega_k)$. Moreover, a task $t_i$ with a higher specific quality $q_i$ will also have a higher probability $P(t_i, \omega_k)$. Then, we have the theorem below.

\begin{theorem}
	We assume $P(t_i, \omega_k)$ is positively related to the difficulty $d_i$ of task $t_i$, which is:
	$$P(t_i, \omega_k)\propto d_i\cdot q_i.$$
	Then, the delay probability $L(t_i)$ of task $t_i$ can be estimated by 
	\begin{equation}
	L(t_i) \propto (d_i\cdot q_i)^{\lceil\frac{\epsilon_{max} -\epsilon_i}{\bar{r}_l}\rceil}\label{eq:delay_possibility}
	\end{equation}
	
	\noindent where $d_i$
	is the difficulty of task $t_i$ given by
	Eq.~(\ref{eq:task_difficulty}) and $\epsilon_i$ is the time lapse of task $t_i$.
\end{theorem}

\begin{proof}
	Task $t_i$ will be delayed if it will not finish in the next $\lceil\frac{\epsilon_{max} -\epsilon_i}{\bar{r}_l}\rceil$ rounds. Then, we have:
	
	\begin{eqnarray}
	L(t_i) =&& P\{t_i \text{ is not finished in round } \omega_{x_1}, \omega_{x_2},..., \omega_{x_{\lceil\frac{\epsilon_{max} -\epsilon_i}{\bar{r}_l}\rceil}} \}\notag\\
	=&& \prod_{k=1}^{\lceil\frac{\epsilon_{max} -\epsilon_i}{\bar{r}_l}\rceil} P\{t_i \text{ is not finished in round } \omega_{x_k} \}\notag\\
	=&& \prod_{k=1}^{\lceil\frac{\epsilon_{max} -\epsilon_i}{\bar{r}_l}\rceil} P(t_i, \omega_{x_k})\notag\\
	\propto&& (d_i\cdot q_i)^{\lceil\frac{\epsilon_{max} -\epsilon_i}{\bar{r}_l}\rceil}\label{eq:indepent_prob}
	\end{eqnarray}
	Eq.~(\ref{eq:indepent_prob}) holds since we assume the probability $P(t_i, \omega_k)$ is positively related to the difficulty and specific quality of  $t_i$. 
\end{proof}

Note that, in this work, we take two major factors: the difficulty and  specific quality of a task, into consideration to build our framework and ignore other minor factors (e.g., spammers and copying workers). We will consider these minor factors as our future work.

\subsubsection{Request-based Scheduling (RBS) Approach}

With the estimation of the delay probabilities of tasks, we propose
a \textit{request-based scheduling} (RBS) approach. In this
approach, when a worker becomes available, he/she will send a
request for the next task to the server. Then, the server calculates
the delay probabilities of the on-going tasks on the platform, and
greedily return the task with the highest delay probability to the
worker.

{
	\begin{algorithm}[th]
		\KwIn{A worker $w_j$ requesting for his/her next task and a set $T= \{t_1, t_2, ..., t_v\}$ of $v$ uncompleted tasks}
		\KwOut{Returned task $t_i$}
		\ForEach{task $t_i$ \textbf{in} $T$} {
			calculate the delay possibility value of $t_i$ with Eq.~(\ref{eq:delay_possibility})\;
		}
		
		select one task $t_i$ with the highest delay probability\;
		
		\uIf{the expected accuracy of $t_i$ is higher than $q_i$}{
			Remove $t_i$ from $T$\;
		}
		
		\Return{$t_i$}\;
		\caption{GreedyRequest($W$, $T$)}
		\label{alg:greedy_request}
	\end{algorithm}
}

The pseudo code of our request-based scheduling approach, namely
{\sf GreedyRequest}, is shown in Algorithm \ref{alg:greedy_request}.
It first calculates the delay probability of each uncompleted task
in $T$ (lines 1-2). Then, it selects a suitable task $t_i$ with the
highest delay probability (line 3). If we find the expected accuracy
of task $t_i$ (given in Eq.~(\ref{eq:task_accuracy})) is higher than
the quality threshold $q_i$, then we will remove task $t_i$ from
$T$. Finally, we return/assign task $t_i$ to worker $w_j$, who is requesting his/her next task.

\noindent \textbf{The Time Complexity of RBS.} We next analyze the
time complexity of the request-based scheduling approach, {\sf
	GreedyRequest}, in Algorithm \ref{alg:greedy_request}. We assume that each task has received $h$ answers. For each task
$t_i$, to compute its difficulty, the time complexity is
$O(h)$. Thus, the time complexity of computing delay probabilities
for all $v$ uncompleted tasks is given by $O(v\cdot h)$ (lines 1-2). Next,
the cost of selecting the task $t_i$ with the highest delay
probability is $O(v)$ (line 3). The cost of checking the
completeness for task $t_i$ and removing it from $T$ is given by
$O(1)$. As a result, the time complexity of our request-based
scheduling approach is given by $O(v\cdot h)$.

\subsubsection{Batch-based Scheduling (BBS) Approach}

Although the RBS approach can easily and quickly respond to each
worker's request, it in fact does not have the control on workers in
this request-and-answer style. Next, we will propose an orthogonal
\textit{batch-based scheduling} (BBS) approach, which assigns each
worker with a list of suitable tasks in a batch, where the length of
the list is determined by his/her response speed.

The intuition of our BBS approach is as follows. If we can assign
high-accuracy workers to difficult and urgent tasks and low-accuracy
workers with easy and not that urgent tasks, then the worker labor
will be more efficient and the throughput of the platform will
increase.

Specifically, in each round, our BBS approach iteratively picks a
task with the highest delay probability (among all the remaining
tasks in the system), and then greedily selects a minimum set of
workers to complete this task. Algorithm \ref{alg:greedy_assign}
shows the pseudo code of the BBS algorithm, namely {\sf
	GreedyBatch}. In particular, since no worker-and-task pair is
assigned at the beginning, we initialize the assignment set
$\mathbb{A}$ as an empty set (line 1). Then, we calculate the delay
probability of each unfinished task (given in
Eq.~(\ref{eq:delay_possibility})) (lines 2-3). Thereafter, we
iteratively assign workers for the next task $t_i$ with the highest
delay probability (lines 4-6). Next, we invoke Algorithm {\sf
	MinWorkerSetSelection}, which selects a minimum set, $W_o$, of
workers who satisfy the required accuracy threshold of task $t_i$
(line 7). If $W_o$ is not empty, then we insert task-and-worker
pairs, $\langle t_i, w_j\rangle$, into set $\mathbb{A}$ (lines
8-10). If each worker $w_i$ cannot be assigned with more tasks, then
we remove him/her from $W$ (lines 11-12). Here, we decide whether a
worker $w_j$ can be assigned with more tasks, according to his/her
response times on categories, his/her assigned tasks, and the round
interval of the BBS approach. That is, if the summation of response
times of the assigned tasks is larger than the round interval, then
the worker cannot be assigned with more tasks; otherwise, we can
still assign more tasks to him/her. 

\begin{algorithm}[h]
	\KwIn{A set, $T= \{t_1, t_2, ..., t_m\}$, of $m$ unfinished tasks and a set, $W=\{w_1, w_2, ..., w_n\}$, of $n$ workers}
	\KwOut{Assignment $\mathbb{A}=\{\langle t_i, w_j\rangle\}$}
	$\mathbb{A} \gets \emptyset$\;
	\ForEach{task $t_i$ \textbf{in} $T$} {
		calculate the delay possibility value of $t_i$ with Eq.~(\ref{eq:delay_possibility})\;
	}
	
	\While{$T \neq \emptyset$ \textbf{and} $W \neq \emptyset$} {
		select task $t_i$ with the highest delay probability value\;
		remove $t_i$ from $T$\;
		$W_o \gets$ MinWorkerSetSelection($t_i$, $W$, $W_i$)\;
		\If{$W_o \neq \emptyset$}{
			\ForEach{$w_j \in W_o$}{
				Insert $\langle t_i, w_j\rangle$ into $\mathbb{A}$\;
				\If{$w_j$ cannot be assigned with more tasks}{
					Remove $w_j$ from $W$\;
				}
			}
		}
	}
	
	\Return{$\mathbb{A}$}\;
	\caption{GreedyBatch($W$, $T$)}
	\label{alg:greedy_assign}
\end{algorithm}

\noindent\textbf{Minimum Worker Set Selection.} In line 7 of
Algorithm 2 above, we mentioned a {\sf MinWorkerSetSelection}
algorithm, which selects a minimum set of workers satisfying the
constraint of the quality threshold $q_i$ for task $t_i$. We will
discuss the algorithm in detail, and prove its correctness below.

Before we provide the algorithm, we first present one property of
the expected accuracy of a task.

\begin{lemma} Given a set of workers, $W_i$, assigned to task $t_i$ in category $c_l$, the expected accuracy of task $t_i$ can be calculated as
	follows:
	\begin{eqnarray}
	\Pr(W_i, c_l) &=& \Pr(W_i-\{w_j\}, c_l) \label{eq:rewrite_expect_accuracy}\\
	&+ &\alpha_{jl}\Big(\sum_{U}\big(\prod_{w_o \in U}\alpha_{ol}\prod_{w_o \in W_i - U - \{w_j\}}(1-\alpha_{ol})\big)\Big)\notag
	\end{eqnarray}
	\noindent where {$U = W_{i,\lceil\frac{k}{2} \rceil} - \{w_j\}$} and { $\Pr(W_i, c_l)$} is defined in Eq.~(\ref{eq:task_accuracy}).
	\label{lemma:rewrite_expect_accuracy}
\end{lemma}
\begin{proof}
	For a task $t_i$ in category $c_l$, assume a set of $k$ workers
	$W_i$ are assigned to it. As the definition of the expected accuracy
	of task $t_i$ in Eq.~(\ref{eq:task_accuracy}) shows, for any subset
	$V'_i \subseteq W_i$ and $|V_i| \geq \lceil\frac{k}{2} \rceil$, when
	worker $w_j \in W_i$ is not in $V_i$, we can find an addend $A$ of
	$$\big(\prod_{w_o \in V_i}\alpha_{ol}\prod_{w_o \in W_i - V_i - {w_j}}(1-\alpha_{ol})\big)\cdot (1-\alpha_{jl})$$
	\noindent in Eq.~(\ref{eq:task_accuracy}). As the
	Eq.~(\ref{eq:task_accuracy}) enumerates all the possible subsets of
	$W_i$ with more than $\lceil\frac{k}{2} \rceil$ elements, we can
	find a subset $V'_i = V_i + \{w_j\}$, which represents another
	addend $A'$ of
	$$\alpha_{jl}\Big(\prod_{w_o \in V'_i-\{w_j\}}\alpha_{ol}\prod_{w_o \in W_i - V'_i}(1-\alpha_{ol})\Big)$$
	\noindent in Eq.~(\ref{eq:task_accuracy}). Then, we have:
	$$A+A' = \prod_{w_o \in V_i}\alpha_{ol}\prod_{w_o \in W_i - V_i - {w_j}}(1-\alpha_{ol}).$$
	
	After we combine all these kind of pairs of addends of worker $w_j$,
	we can obtain:
	\begin{eqnarray}
	&&\Pr(W_i, c_l) \notag\\
	&=& \sum_{x = \lceil\frac{k}{2} \rceil -1}^{k-1}\sum_{W'_{i,x}}\Big(\prod_{w_o \in W'_{i,x}}\alpha_{ol}\prod_{w_o \in W'_i - W'_{i,x}}(1-\alpha_{ol})\Big) \notag
	\\&&+\alpha_{jl}\Big(\sum_{U}\big(\prod_{w_o \in U}\alpha_{ol}\prod_{w_o \in W_i - U - \{w_j\}}(1-\alpha_{ol})\big)\Big)\notag\\
	&=&\Pr(W'_i, c_l) \label{eq:accuracy_present}\\
	&&+\alpha_{jl}\Big(\sum_{U}\big(\prod_{w_o \in U}\alpha_{ol}\prod_{w_o \in W_i - U - \{w_j\}}(1-\alpha_{ol})\big)\Big),\notag
	\end{eqnarray}

	\noindent where $W'_i = W_i - \{w_j\}$ and $W'_{i,x}$ is a subset of
	$W'_i$ with $x$ elements, and $U=W_{i,\lceil\frac{k}{2} \rceil} -
	\{w_j\}$. Eq.~(\ref{eq:accuracy_present}) holds as $k$ is always odd
	to ensure the majority voting can get a final result. Note that, the accuracy $\alpha_{jl}$ of worker $w_j$ towards task category $c_l$ is larger than 0.5. The reason is when $\alpha_{jl}$ is smaller than 0.5, we can always treat the answer of worker $w_j$ to $c_l$ as the opposite answer, then the accuracy may become $\alpha'_{jl} = 1 - \alpha_{jl} > 0.5$.
\end{proof}

We can derive two corollaries below.
\begin{corollary} For a task $t_i$ in category $c_l$ with a set of $k$ assigned workers $W_i$, if the category accuracy $\alpha_{jl}$ of any worker $w_j \in W_i$ increases, the expected accuracy $\Pr(W_i, c_l)$ of task $t_i$ will increase (until reaching 1).
	\label{coro:accuracy_increaes}
\end{corollary}

\begin{proof}
	In Eq.~(\ref{eq:rewrite_expect_accuracy}), when the accuracy
	$\alpha_{jl}$ of worker $w_j$ increases,
	the first factor $\Pr(W_i-\{w_j\}, c_l)$ will not be affected, and the second factor will increase. Note that, when all the workers are 100 $\%$ accurate, $\Pr(W_i-\{w_j\}, c_l) = 1$ and the second factor equals to 0, which leads to that the expected accuracy stays at 1. Thus, the corollary is proved.
\end{proof}

\begin{corollary} For a task $t_i$ in category $c_l$ with a set of $k$ assigned workers $W_i$, if we assign a new worker $w_j$ to task $t_i$, the expected accuracy of task $t_i$ will increase.
	\label{coro:worker_increaes}
\end{corollary}

\begin{proof}
	With Lemma \ref{lemma:rewrite_expect_accuracy}, we can see that when adding one more worker to a task $t_i$, the expected accuracy of task $t_i$ will increase. In Eq. (\ref{lemma:rewrite_expect_accuracy}), the first factor $\Pr(W_i-\{w_j\}, c_l)$ is the expected accuracy of task $t_i$ before adding worker $w_j$. The second factor is larger than 0 as the accuracy $\alpha_{jl}$ of worker $w_j$ towards task category $c_l$ is larger than 0.5. When $\alpha_{jl}$ is smaller than 0.5, we can always treat the answer of worker $w_j$ to $c_l$ as the opposite answer, then the accuracy  becomes $\alpha'_{jl} = 1 - \alpha_{jl} > 0.5$ .
\end{proof}

With Corollaries \ref{coro:accuracy_increaes} and \ref{coro:worker_increaes}, to increase the expected accuracy of a task $t_i$, we can use workers with higher category accuracies or assign more workers to task $t_i$. When the required expected accuracy of a task $t_i$ is given, we can finish task $t_i$ with a smaller number of high-accuracy workers.
To accomplish
as many tasks as possible, we aim to greedily pick the least number
of workers to finish each task iteratively.

\begin{algorithm}
	\KwIn{A set $W=\{w_1, w_2, ..., w_n\}$ of available workers, a task $t_i$ in category $c_l$ with a set of already assigned workers $W_i$}
	\KwOut{A minimum set of workers assigned to task $t_i$}
	$W_o \gets W_i$\;
	
	\While{$Pr(W_o, c_l) < q_i$ and $|W - W_o| > 0$} {
		choose a new worker $w_j$ with the highest accuracy $\alpha_{jl}$;
		$W$.remove($w_j$)\;
		$W_o$.add($w_j$)\;
	}
	\uIf{$Pr(W_o, c_l) \geq q_i$}{
		\Return{$W_o - W_i$}\;
	}
	\Else{
		\Return{$\emptyset$}\;
	}
	
	\caption{MinWorkerSetSelection($t_i$, $W$, $W_i$)}
	\label{alg:min_worker_selection}
\end{algorithm}

Algorithm \ref{alg:min_worker_selection} exactly shows the procedure
of {\sf MinWorkerSetSelection}, which selects a minimum set, $W_o$,
of workers to conduct task $t_i$. In each iteration, we greedily
select a worker $w_j$ (who has not been assigned to task $t_i$) with
the highest accuracy in the category of task $t_i$, and assign
workers to task $t_i$ (lines 2-4). If such a minimum worker set
exists, we return the newly assigned worker set; otherwise, we
return an empty set (lines 5-8). The correctness of Algorithm
\ref{alg:min_worker_selection} is shown below.

\begin{lemma} The number of workers in the set $W_o$ returned by Algorithm \ref{alg:min_worker_selection} is minimum, if $W_o$ exists.
	\label{lemma:greedy_accurate_worker}
\end{lemma}

\begin{proof}
	Let set $W_o$ be the returned by Algorithm \ref{alg:min_worker_selection}
	to satisfy the quality threshold $q_i$ and worker $w_j$ is the last one
	added to set $W_o$.  Assume there is a subset of workers $W' \subseteq W$ such that
	$|W'| = |W_o| - 1$ and $Pr(W', c_l) \geq q_i$.
	
	Since each worker in $W_o$ is greedily picked with the highest current
	accuracy in each iteration of lines 2-4 in Algorithm \ref{alg:min_worker_selection},
	for any worker, $w_k \in W_o$ will have higher accuracy than any
	worker $w'_k \in W'$. As $|W'| = |W_o - \{w_j\}|$, according to
	Corollary \ref{coro:accuracy_increaes}, $Pr(W_o-\{w_j\}, c_l) > Pr(W', c_l)$.
	However, as $w_j$ is added to $W_o$, it means $Pr(W_o-\{w_j\}, c_l) < q_i$.
	It conflicts with the assumption that $Pr(W', c_l) \geq q_i$. Thus, set $W'$ cannot exist.
\end{proof}

\noindent \textbf{The Time Complexity of BBS.} To analyze the time
complexity of the batch-based scheduling (BBS) approach, called {\sf
	GreedyBatch}, as shown in Algorithm \ref{alg:greedy_assign}, we
assume that each task $t_i$ needs to be answered by $h$ workers. The
time complexity of calculating the delay probability of a task is
given by $O(m\cdot h)$ (lines 2-3). Since each iteration solves one task,
there are at most $m$ iterations (lines 4-13). In each iteration,
selecting one task $t_i$ with the highest delay probability requires
$O(m)$ cost (line 5). The time complexity of the {\sf
	MinWorkerSetSelection} procedure is given by $O(h)$ (line 7). The
time complexity of assigning workers to the selected task $t_i$ is
$O(h)$ (lines 8-12). Thus, the overall time complexity of the BBS
approach is given by $max(O(m^2), O(m\cdot h))$.

\section{The Notification Module}
\label{sec:notification}

In this section, we introduce the detailed model of the notification
module in our PROG framework (as mentioned in Section
\ref{sec:frame}), which is in charge of sending invitation
notifications to offline workers in order to maintain enough online
workers doing tasks. Since it is not a good idea to broadcast to all
offline workers, our notification module only sends notifications to
those workers with high probabilities of accepting invitations.

\subsection{Kernel Density Estimation for Worker Availability}

In this subsection, we will model the availability of those
(offline) workers from historical records. The intuition is that,
for each worker, the patten of availability on each day is
relatively similar. For example, a worker may have the spare time to
do tasks, when he/she is on the bus to the school (or company) at
about 7 am every morning. Thus, we may obtain their historical data
about the timestamps they conducted tasks.

However, the number of historical records (i.e., sample size) for
each worker might be small. In order to accurately estimate the
probability of any timestamp that a worker is available,  we use a non-parametric approach, called \textit{kernel
	density estimation} (KDE) \cite{rosenblatt1956remarks}, based on
random samples (i.e., historical timestamps that the worker is
available).

Specifically, for a worker $w_j$, let $E_j=\{e_1, e_2, ..., e_n\}$
be a set of $n$ active records that worker $w_j$ did some tasks,
where event $e_i$ ($1\leq i\leq n$) occurs at timestamp $ts_i$.
Then, we can use the following KDE estimator to compute the
probability that worker $w_j$ is available at timestamp $ts$: 
\begin{equation}
f(ts | E_j, h) =
\frac{1}{nh}\sum_{i=1}^{n}K\left(\frac{ts-ts_i}{h}\right),\notag
\label{eq:kde}
\end{equation}

\noindent where $e$ is the event that worker $w_j$ is available and
will accept the invitation at a given time\-stamp $ts$, $K(\cdot)$
is a kernel function (here, we use Gaussian kernel function $K(u) =
\frac{1}{\sqrt{2\pi}}e^{-u^2/2}$), and $h$ is a scalar bandwidth
parameter for all events in $E$. The bandwidth of the kernel is a
free parameter and exhibits a strong influence on the estimation.
For simplicity, we set the bandwidth following a rule-of-thumb
\cite{silverman1986density} as follows: 
\begin{equation}
h = \Big( \frac{4 \hat{\sigma}^{5}}{3n}\Big)^{\frac{1}{5}} =
1.06\hat{\sigma}n^{-1/5}, \label{eq:bandwidth}
\end{equation}

\noindent where $\hat{\sigma}$ is the standard deviation of the
samples. The rule works well when density is close to being normal,
which is however not true for estimating the probability of workers
at a given time\-stamp $ts$. However, adapting the kernel
bandwidth $h_i$ to each data sample $e_i$ may overcome this issue
\cite{breiman1977variable}.

Inspired by this idea, we select $k$ nearest neighbors of event
$e_i$ (here, we consider neighbors by using time as measure, instead
of distance), and calculate the adaptive bandwidth $h_i$ of event
$e_i$ with $(k+1)$ samples using Eq.~(\ref{eq:bandwidth}), where $k$
is set to $\beta \cdot n$ ($\beta$ is a ratio parameter).
Afterwards, we can define the adaptive bandwidth KDE as follows:
\begin{equation}
f(ts| E_j) =
\frac{1}{n}\sum_{i=1}^{n}K\left(\frac{ts-ts_i}{h_i}\right).
\label{eq:adaptive_kde}
\end{equation}

\subsection{Smooth Estimator}

Up to now, we have discussed the adaptive kernel density approach to
estimate the probability that a worker is available, based on one's
historical records (samples). However, some workers may just
register or rarely accomplish tasks, such that his/her historical
events are not available or enough to make accurate estimations,
which is the ``cold-start'' problem that often happens in the
recommendation system \cite{schein2002methods}.

Inspired by techniques \cite{schein2002methods} used to solve such a
cold-start problem in recommendation systems and the influence among
friends \cite{cho2011friendship} (i.e., friends tend to have similar
behavior patterns, such as the online time periods), we propose a
\textit{smooth KDE model} (SKDE), which combines the individual's
kernel density estimator with related scale models. That is, for
each worker, we can use historical data of his/her friends to
supplement/predict his/her behaviors.

Here, our FROG platform is assumed to have the access to the
friendship network of each worker, according to his/her social
networks (such as Facebook, Twitter, and WeChat). In our experiments
of this paper, our FROG platform used data from the WeChat network.

Specifically, we define a smooth kernel density estimation model as
follows:
\begin{equation}
P_{SKDE}(ts | E_j, E) = \sum_{s=1}^{S}\alpha_s f(ts |E^s)
\label{eq:skde}
\end{equation}
\noindent where $\alpha_1, \alpha_1, ..., \alpha_s$ are non-negative
\textit{smoothing factors} with the property of
$\sum_{s=1}^S\alpha_s = 1$, $E$ is the entire historical events of
all the workers, and $f(ts|E^s)$ is the $s$-th \textit{scaling
	density estimator} calculated on the subset events $E^s$.

For a smooth KDE model with $S$ ($>2$) scaling density estimators,
the first scaling density estimator can be the basic individual
kernel density estimator with $E^1=E_j$ and the $S$-th scaling
density estimator can be the entire population density estimator
with $E^s = E$. Moreover, since our FROG platform can obtain the
friendship network of each worker (e.g., Facebook, Twitter, and
WeChat), after one registers with social media accounts, we can find
each worker's $k$-step friends. This way, for the intermediate
scaling density estimators $s=2, ..., S-1$, we can use different
friendship scales, such as the records of the 1-step friends, 2-step
friends, ..., $(S-2)$-step friends of worker $w_i$.
According to the famous Six degrees of separation theory
\cite{barabasi2002linked}, $S$ is not larger than 6. However, in practice, we in fact can only use 1-step or 2-step friends, as the intermediate scaling density estimators may involve too many workers of when $S$ is too large. Alternatively, other relationship can also be used to smooth the KDE model, such as the location information of workers. One possible variant is to classify the workers based on their locations, as workers in close locations may work or study together such that their time schedules may be similar with each other.

To train the SKDE model, we need to set proper values for smoothing
factors $\alpha_s$. We use the latest event records as validation
data $E^v$ (here $|E^v|=n$), and other history records as the
training data $E^h$. Specifically, for each event $e_k$ in $E^v$, we
have the estimated probability as follows:
\begin{equation}
P(ts_k|E^h, \alpha) = P_{SKDE}(ts_k | E^h, \alpha) = \sum_{s = 1}^{S}\alpha_s f(ts_k|E^s)\notag
\end{equation}

\noindent where $S$ is the number of scaling density estimators.
Then, to tune the smoothing factors, we use the Maximum Likelihood
Estimation (MLE) with log-likelihood as follows:
\begin{equation}
\hat{\alpha} = argmax_{\alpha} \log \big(\prod_{k=1}^{n} P(ts_k|E^h, \alpha)\big)
\label{eq:mle}
\end{equation}

However, Eq.~(\ref{eq:mle}) is not trivial to solve, thus, we use EM
algorithm to calculate its approximate result.

We initialize the smoothing factors as $\alpha_s = 1/S$ for $s=1, 2,
..., S$. Next, we repeat Expectation-step and Maximization-step,
until the smoothing factors converge.

\noindent \textbf{Expectation Step.} We add a latent parameter
$Z=\{z_1, z_2, ..., z_S\}$, and its distribution on $ts_k$ is
$Q_k(Z)$, then we can estimate $Q_k(Z)$ as follows:
\begin{eqnarray}
Q^t_k(z_s|\alpha^t) &=& P(z_s|ts_k, E^h, \alpha^t) \notag\\
&=& \frac{\alpha^t_s \cdot f(ts_k|E^s)}{\sum_{i=1}^{S}\alpha^t_i\cdot
	f(ts_k|E^i)},\notag \label{eq:exp_step}
\end{eqnarray}

\noindent where $f(\cdot)$ is calculated with Eq. (\ref{eq:adaptive_kde}).

\noindent \textbf{Maximization Step.} Based on the expectation
result of the latent parameter $Z$, we can calculate the next
smoothing factor values $\alpha^{t+1}$ with MLE as follows:
\begin{eqnarray}
\alpha^{t+1} &=& argmax_{\alpha} \sum_{k=1}^{n}\sum_{s=1}^{S} Q^t_k(z_s)\log\Big(\frac{P(st_k, z_s|\alpha^t)}{Q^t_k(z_s)}\Big) \notag\\
&=&argmax_{\alpha} \sum_{k=1}^{n}\sum_{s=1}^{S}
Q^t_k(z_s)\log\Big(\frac{\alpha^t_s \cdot
	f(ts_k|E^s)}{Q^t_k(z_s)}\Big),\notag \label{eq:max_step}
\end{eqnarray}
\noindent where $f(\cdot)$ is calculated with Eq. (\ref{eq:adaptive_kde}).

\subsection{Solving of the Efficient Worker Notifying Problem}

As given in Definition \ref{definition:notification}, our EWN
problem is to select a minimum set of workers with high
probabilities to accept invitations, to whom we will send
notifications.

Formally, given a trained smooth KDE model and a timestamp $ts$,
assume that we want to recruit $u$ more workers for the FROG
platform. In the EWN problem (in Definition
\ref{definition:notification}), the acceptance probability
$P_{ts}(w_j)$ of worker $w_j$ can be estimated by
Eq.~(\ref{eq:skde}).

Next, with Definition \ref{definition:worker_dominance}, we can 
sort workers, $w_j$, based on their ranking scores $R(w_j)$
(e.g., the number of workers dominated by each worker)
\cite{yiu2007efficient}. Thereafter, we will notify top-$v$ workers
with the highest ranking scores.

The pseudo code of selecting worker candidates is shown in Algorithm
\ref{alg:worker_notify_selection}. We first initialize the selected
worker set, $W_n$, with an empty set (line 1). Next, we calculate
the ranking scores of each worker (e.g., the number of other workers can be dominated with the Definition \ref{definition:worker_dominance}) (lines 2-3). Then, we iteratively
pick workers with the highest ranking scores until the selected
workers are enough or all workers have been selected (lines 4-8).
Finally, we return the selected worker candidates to send invitation
notifications (line 9).
\begin{algorithm}[h]
	\KwIn{A set, $W=\{w_1, w_2, ..., w_n\}$, of offline workers, the expected number, $u$, of acceptance workers, and  the current timestamp $ts$}
	\KwOut{A set, $W_n$, of workers to be invited}
	$W_n = \emptyset$\;
	\ForEach{worker $w_j$ \textbf{in} $W$} {
		calculate the ranking score $R(w_j)$ of $w_j$\;
	}
	
	\While{$u > 0$ \textbf{and} $|W| >0$}{
		select one worker $w_j$ with the highest ranking score in $W$\;
		$W = W-\{w_j\}$\;
		$W_n$.add($w_j$)\;
		$u=u-P_{ts}(w_j)$
	}
	
	\Return{$W_n$}\;
	\caption{WorkerNotify($W$, $T$)}
	\label{alg:worker_notify_selection}
\end{algorithm}

\noindent \textbf{The Time Complexity.} To compute the ranking
scores, we need to compare every two workers, whose time complexity
is $O(n^2)$. In each iteration, we select one candidate, and there
are at most $n$ iterations. Assuming that $n$ workers are sorted by
their ranking scores, lines 4-8 have the time complexity $O(n \cdot
\log(n))$. Thus, the time complexity of Algorithm
\ref{alg:worker_notify_selection} is given by $O(n^2)$.

\noindent \textbf{Discussions on Improving the EWN Efficiency.} To improve the efficiency of calculating the ranking scores of workers, we may utilize a 3D grid index to accelerate the computation, where 3D includes the acceptance probability, response time, and accuracy. 
Each worker is in fact a point in a 3D space w.r.t. these 3 dimensions. If a worker $w_j$ dominates a grid cell $gc_x$, then all workers in cell $gc_x$ are dominated by $w_j$. Similarly, if worker $w_j$ is dominated by the cell $gc_x$, then all the workers in $gc_x$ cannot be dominated by $w_j$. Then, we can compute the lower/upper bounds of the ranking score for each worker, and utilize them to enable fast pruning \cite{yiu2007efficient}.

\section{Experimental Study}
\label{sec:exper}

\subsection{Experimental Methodology}

\noindent \textbf{Data Sets for Experiments on Task Scheduler
	Module.} We use both real and synthetic data to test our task
scheduler module. We first conduct a set of comparison experiments
on the real-world crowdsourcing platform, gMission \cite{chen2014gmission}, where workers do tasks and are notified via WeChat \cite{wechat}, and
evaluate our task scheduler module on 5 data sets \cite{dataeveryone}. Tasks in each
data set belong to the same category. For each experiment on the real platform, we use 100
tasks for each data set (category). We manually label the ground
truth of tasks. To subscribe one category, each worker is required
to take a qualification test consisting of 5 testing questions. We uniformly
generate quality threshold for each task within the range [0.8,
0.85]. 
Below, we give brief descriptions of the 5 real data sets.

\noindent 1) \textit{Disaster Events Detection (DED)}: DED contains
a set of tasks, which ask workers to determine whether a tweet
describes a disaster event. For example, a task can be ``Just
happened a terrible car crash'' and workers are required to select
``Disaster Event'' or ``Not Disaster Event''.

\noindent 2) \textit{Climate Warming Detection (CWD)}: CWD is to
determine whether a tweet considers the existence of global
warming/climate change or not. The possible answers are ``Yes'', if the
tweet suggests global warming is occurring; otherwise, The possible answers are ``No''. One tweet example is
``Global warming. Clearly.'', and workers are expected to answer
``Yes''.

\noindent 3) \textit{Body Parts Relationship Verification (BPRV)}:
In BPRV, workers should point out if certain body parts are part of
other parts. Questions were phrased like: ``[Part 1] is a part of
[part 2]''. For example, ``Nose is a part of spine'' or ``Ear is a
part of head.'' Workers should say ``Yes'' or ``No'' for this
statement.

\noindent 4) \textit{Sentiment Analysis on Apple Incorporation
	(SAA)}: Workers are required to analyze the sentiment about Apple,
based on tweets containing ``\#AAPL, @apple, etc''. In each task,
workers are given a tweet about Apple, and asked whether the user is
positive, negative, or neutral about Apple. We used records with
positive or negative attitude about Apple, and asked workers to
select ``positive'' or ``negative'' for each tweet.

\noindent 5) \textit{App Search Match (ASM)}: In ASM, workers are
required to view a variety of searches for mobile Apps, and
determine if the intents of those searches are matched. For example,
one short query is ``music player''; the other one is a longer one
like ``I would like to download an App that plays the music on the
phone from multiple sources like Spotify and Pandora and my
library.'' If the two searches have the same intent, workers should
select ``Yes''; otherwise, they should select ``No''.

For synthetic data, we simulate crowd workers based on the
observations from real platform experiments. Specifically, in
experiments on the real platform, we measure the average response
time, $\bar{r_{jl}}$, of worker $w_j$ on category $c_l$, the variance of
the response time $\sigma_{jl}^2$, and the category accuracy
$\alpha_{jl}$. Then, to generate a worker $w'_j$ in the synthetic data set, we first randomly select one worker $w_j$ from the workers in the real platform experiments, and produce his/her
response speed $r'_{jl}$ on category $c_l$ following a Gaussian distribution $r'_{jl} \sim \mathcal{N} (\bar{r_{jl}},\sigma_{jl}^2)$, where $\bar{r_{jl}}$ and $\sigma_{jl}^2$ are the average and variance of the response time of worker $w_j$. In addition, we initial the category accuracy $\alpha'_{jl}$ of worker $w'_j$ as that of the worker $w_j$. Moreover, we uniformly generate required number of tasks by sampling from the real tasks. 
Table \ref{tab:experiment} depicts the parameter settings in our
experiments on synthetic datasets, where default values of parameters are in bold font. In
each set of experiments, we vary one parameter, while setting other
parameters to their default values.

\begin{table}[t!]
	\parbox[b]{\linewidth}{
		\caption{ Experimental Settings.} \label{tab:experiment}
		\begin{tabular}{l|l}
			{\bf \qquad \qquad \quad Parameters} & {\bf \qquad \qquad \qquad Values} \\ \hline \hline
			the number of categories $l$& 5, 10, \textbf{20}, 30, 40\\
			the number of tasks $m$ & 1000, 2000, \textbf{3000}, 4000, 5000 \\
			the number of workers $n$ & 100, 200, \textbf{300}, 400, 500 \\
			the range of quality &[0.75, 0.8], \textbf{[0.8, 0.85]}, [0.85, 0.9], [0.9, 0.95]\\threshold $[q^-, q^+]$ \\
			\hline
		\end{tabular}
	}
\end{table}

\noindent \textbf{Data Sets for Experiments on Notification Module.}
To test our notification module in the FROG framework, we utilize
\textit{Higgs Twitter Dataset} \cite{de2013anatomy}. The Higgs
Twitter Dataset is collected for monitoring the spreading process on
the Twitter, before, during, and after the announcement of the
discovery of a new particle with features of the elusive Higgs boson
on July 4th, 2012. The messages posted on the Twitter about this
discovery between July 1st and 7th, 2012 are recorded. There are
456,626 user nodes and 14,855,842 edges (friendship connections) between them. In addition,
the data set contains 563,069 activities. Each activity happens
between two users and can be retweet, mention, or reply. We
initialize the registered workers on our platform with users in the
Higgs Twitter Dataset (and their relationship on the Twitter). What
is more, the activities in the data set is treated as online records
of workers on the platform. The reason is that only when a user is free, he/she can make activities on Twitter.

\noindent \textbf{Competitors and Measures.} For the task scheduler
module, we conduct experiments to test our two adaptive scheduling
approaches, request-based (RBS) and batch-based scheduling (BBS)
approaches. We select the task assigner of iCrowd framework \cite{fan2015icrowd} as a competitor (iCrowd), which iteratively resolves a task with a set of $k$ available workers having the maximum average accuracy in the current situation. Here $k$ is a requester-specified parameter and we configure it to 3 following its setting in \cite{fan2015icrowd}. In addition, we compare them with a random method, namely RANDOM,
which randomly routes tasks to workers, and a fast-worker greedy method, namely fGreedy, which greedily pick the fastest workers to finish the task with the highest delay possibility value. We hire 70 workers from the WeChat platform to conduct the experiments. Table \ref{tab:statistics} shows the statistics of category accuracies and category response times of top 5 workers, who conducted  the most tasks.

For the notification module, we conduct experiments to compare our
smooth KDE model with our KDE model without smoothing and a random method, namely Random, which
randomly selects workers. Moreover, we also compare our approach with a simple
method, namely \textit{Nearest Worker Priority} (NWP), which selects
workers with the most number of historical records within the
$\theta$-minute period before or after the given timestamp in previous
dates. Here, we use $\theta =15$, as it is sufficient for a worker to response the invitation. For each predicted worker, if he/she has activities within the time period  from the the target timestamp to 15 minutes later, we treat that it is a correct prediction. At timestamp $ts$, we denote $N_c(ts)$ as the number of correct predictions, $N_t(ts)$ as the number of total predictions and $N_a(ts)$ as the number of activities that really happened.

\begin{table}[t!]
	\parbox[b]{\linewidth}{
		\caption{ Statistics of Workers.} \label{tab:statistics}
		
		\begin{tabular}{c|l|l|l|l|l}
			\multirow{2}{*}{{\bf ID} }& \multicolumn{5}{c|}{{\bf Category Accuracy / Response Time}} \\ 
			& \multicolumn{1}{c}{DED} & \multicolumn{1}{c}{CWD} &\multicolumn{1}{c}{BPRV} &\multicolumn{1}{c}{SAA}&\multicolumn{1}{c|}{ASM}\\\hline \hline
			42 & 0.90/17.78 & 0.91/13.12 & 0.96/4.56 & 0.96/11.45 & 0.87/10.33\\\hline
			57 & 0.94/21.25 & 0.94/14.52& 0.99/4.41 & 0.97/13.82 & 0.92/12.48\\\hline
			134 & 0.78/15.79 & 0.83/10.51& 0.94/5.15& 0.87/11.97 & 0.89/11.29\\\hline
			153 & 0.65/24.06 & 0.74/12.08& 0.63/8.53 & 0.91/16.75 & 0.87/ 9.79\\\hline
			155 & 0.83/19.97 & 0.95/13.04 & 0.92/5.03 & 0.88/7.37 & 0.93/14.38\\\hline
			\hline
		\end{tabular}
	}
\end{table}

For experiments on the task scheduler module, we report maximum
latencies of tasks and average task accuracies, for both our
approaches and the competitor method. We also evaluate the final results through the Dawid and Skene's expectation maximization method \cite{dawid1979maximum, ipeirotis2010quality}. Due to space limitation, please refer to \textbf{Appendix B} of our supplemental materials for more details. For experiments on the
notification module, we present the precision ($=\frac{N_c(ts)}{N_t(ts)}$) and recall ($=\frac{N_c(ts)}{N_a(ts)}$) of all tested
methods. Our experiments were run on an Intel Xeon X5675 CPU
with 32 GB RAM in Java.

\begin{figure}[t!]
	\centering
	\subfigure[][{ Maximum Latency}]{
		\scalebox{0.2}[0.2]{\includegraphics{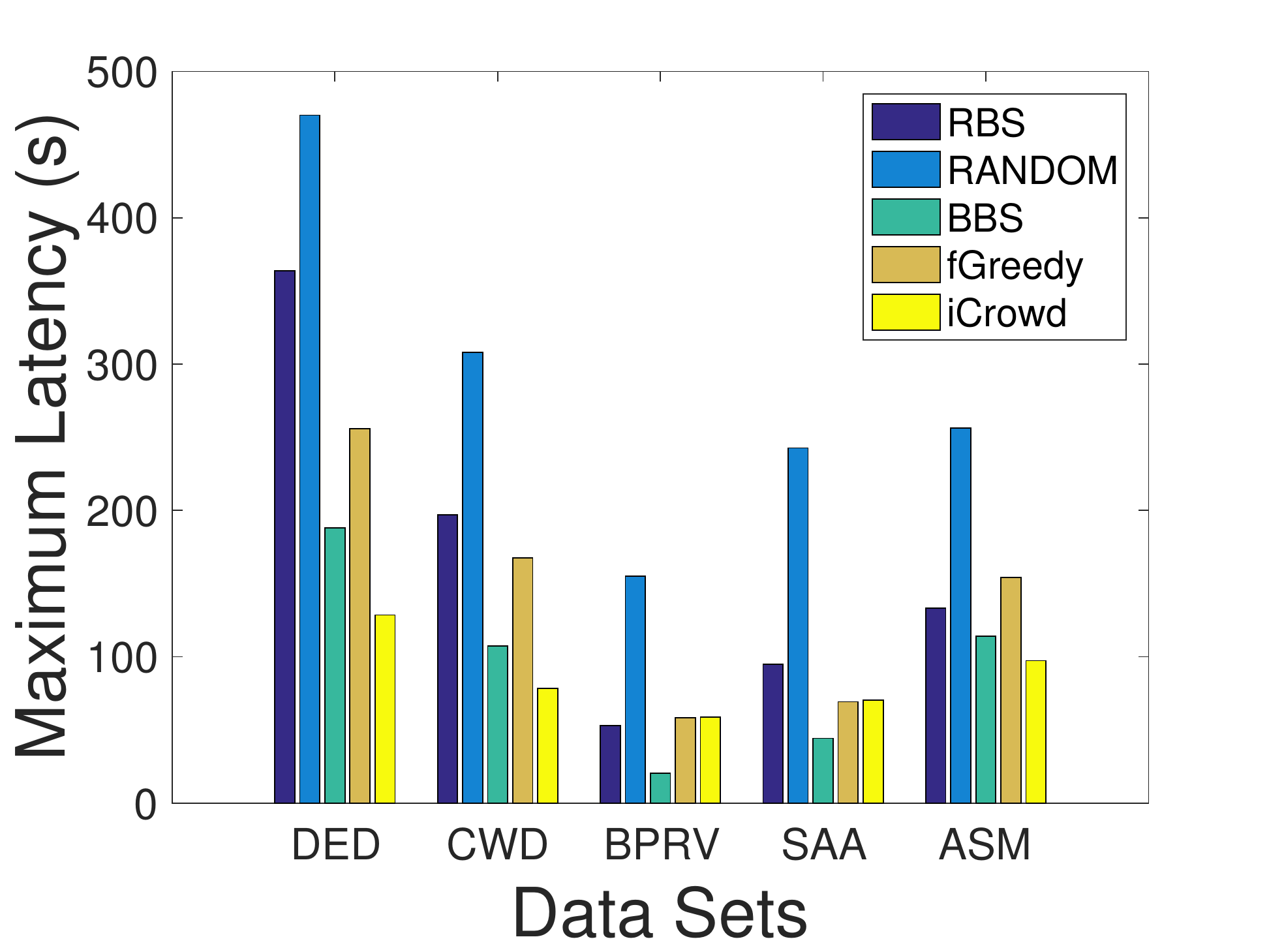}}
		\label{subfig:real_delay}}
	\subfigure[][{ Average Accuracy}]{
		\scalebox{0.2}[0.2]{\includegraphics{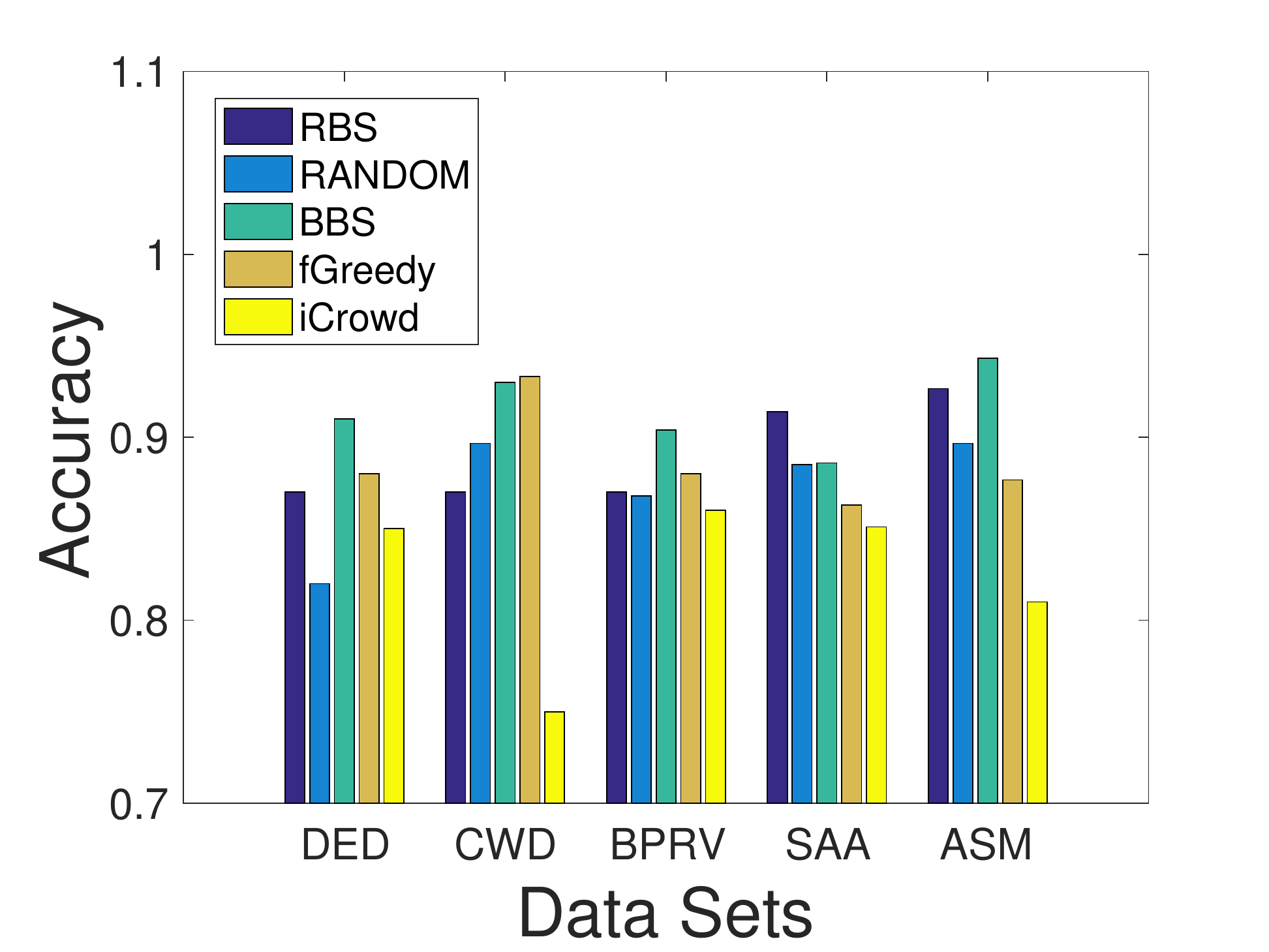}}
		\label{subfig:real_accuracy}}
	\caption{ The Performance of Task Scheduler Module on Real Data.}
	\label{fig:real_fig}
\end{figure}

\subsection{Experiments on Real Data}
\label{sec:exp_real}

\begin{figure}[t!]
	\centering
	\subfigure[][{ Recall}]{
		\scalebox{0.2}[0.2]{\includegraphics{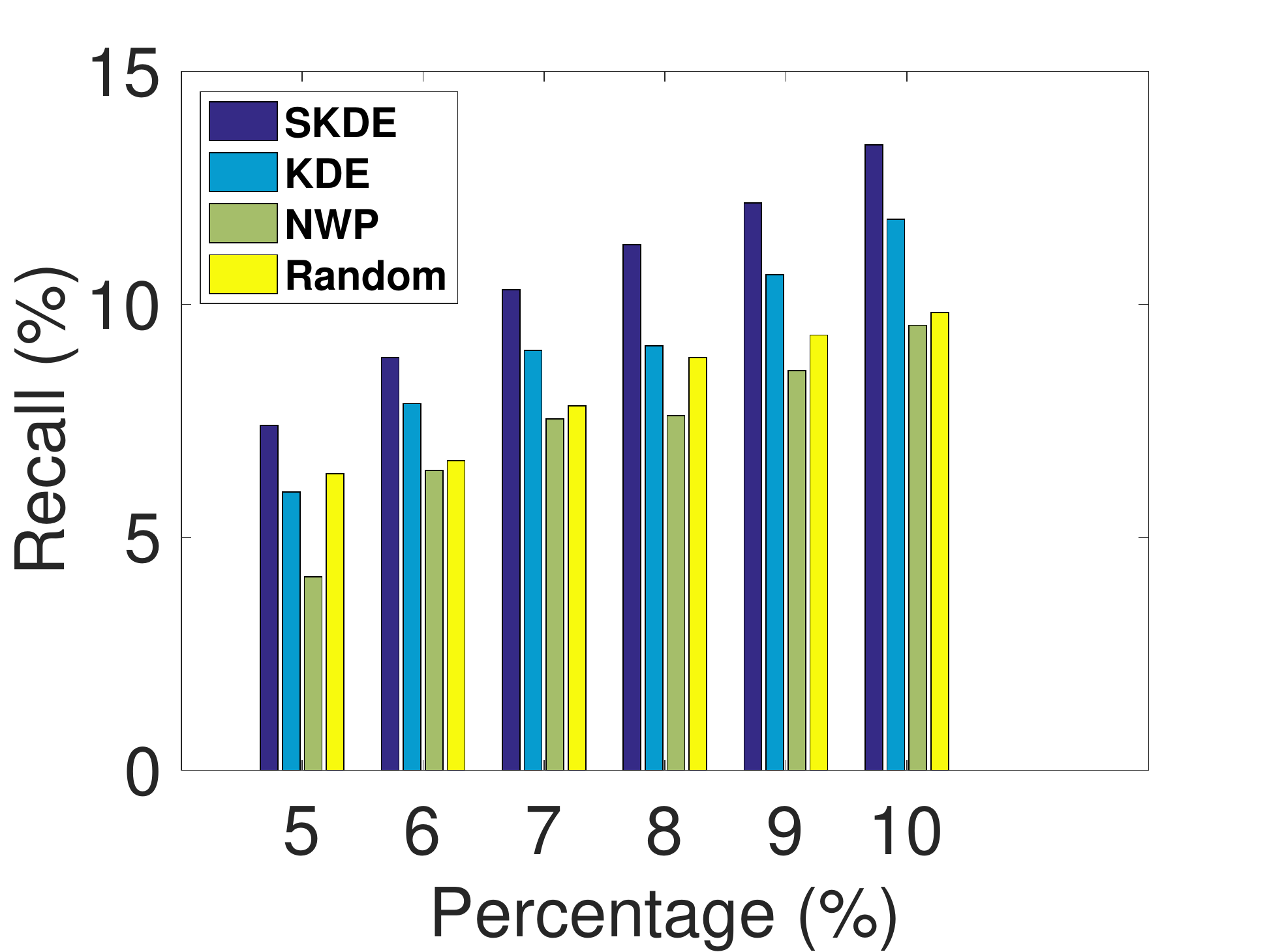}}
		\label{subfig:notify_recall}}
	\subfigure[][{ Precision}]{
		\scalebox{0.2}[0.2]{\includegraphics{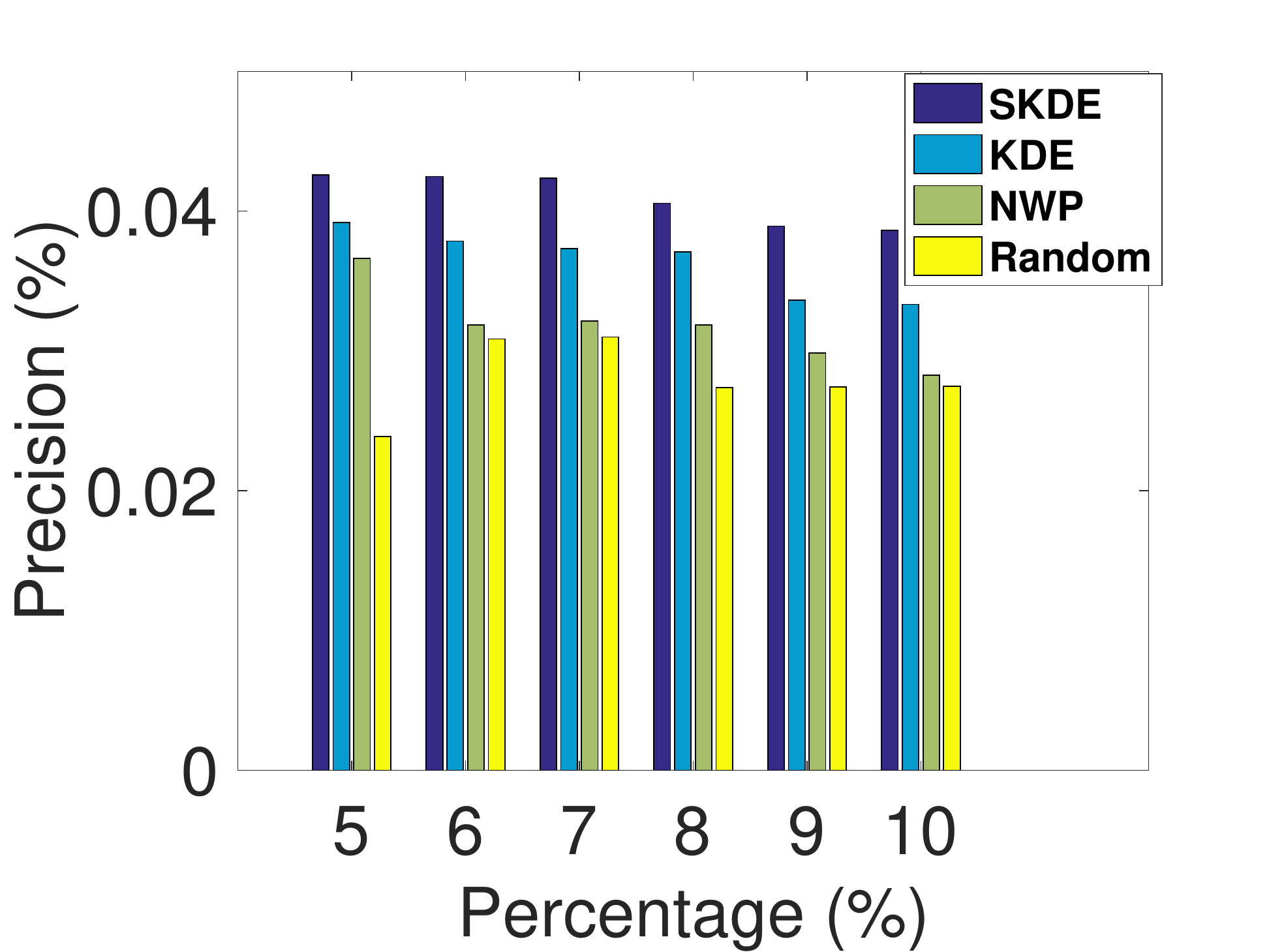}}
		\label{subfig:notify_precision}}
	\caption{ The Performance of the Notification Module on Real Data.}
	\label{fig:notify_fig}
\end{figure}

\noindent \textbf{The Performance of the Task Scheduler Module on Real Data.}
Figure \ref{fig:real_fig} shows the results of experiments on our real platform about the task scheduler module of our framework. For the maximum latencies shown in Figure \ref{subfig:real_delay}, our two approaches can maintain lower latencies than the baseline approach, RANDOM. Specifically, BBS can achieve a much lower latency, which is at most half of that of RANDOM. fGreedy is better than RANDOM, however, still needs more time to finish tasks than our BBS. As iCrowd assigns $k$ (=3) workers to each task, it can achieve lower latencies than out BBS in DED, CWD and ASM but higher latencies than our BBS in BPRV and SAA. For the accuracies shown in Figure \ref{subfig:real_accuracy}, our two approaches achieve higher accuracies than RANDOM. Moreover, the accuracy of BBS is higher than that of RBS. The reason is that, BBS can complete the
most urgent tasks with minimum sets of workers, achieving the
highest category accuracies. In contrast, RBS is not concerned with
the accuracy, and just routes available workers to tasks with the
highest delay probabilities. Thus, RBS is not that effective,
compared with BBS, to maintain a low latency. As the required accuracies are satisfied when assigning tasks to workers, BBS, RBS, RANDOM and fGreedy achieve close accuracies to each other. However, iCrowd just achieves relatively low accuracy in CWD as it assigns only $k$ (=3) workers to each task  and the average accuracy of workers in CWD is low. 

\noindent \textbf{The Performance of Notification Module on Real Data.} To show the effectiveness of our smooth KDE model, we present the recall and precision of our model compared with {\sf KDE}, {\sf NWP} and  {\sf Random}, by varying the number of prediction samples from 5\% to 10\% of the entire population. As shown in Figure \ref{subfig:notify_recall}, our smooth KDE model can achieve higher recall scores than the other three baseline methods. In addition, when we predict with more samples, the advantage of our smooth KDE model is more obvious w.r.t. the recall scores. The reason is that our smooth KDE model can utilize the influence of the friends, which is more effective when we predict with more samples. Similarly, in Figure \ref{subfig:notify_precision}, smooth KDE model can obtain the highest precision scores among all tested methods.

\subsection{Experiments on Synthetic Data}
\label{sec:exp_synthetic}

\noindent \textbf{Effect of the Number, $m$, of Tasks.} Figure
\ref{fig:task_fig} shows the maximum latency and average accuracy of
five approaches, RBS, BBS, iCrowd, RANDOM and fGreedy, by varying the number, $m$,
of tasks from $1K$ to $5K$, where other parameters are set to their
default values. As shown in Figure \ref{subfig:task_latency}, with
more tasks (i.e., larger $m$ values), all the five approaches
achieve higher maximum task latency. This is because, if there are
more tasks, each task will have relatively fewer workers to assign,
which prolongs the latencies of tasks.
RANDOM always has higher latency than our RBS approach, followed by
BBS. fGreedy can achieve lower latency than RBS approach, but still higher than BBS, as fGreedy is still a batch-based algorithm but greedily picking fastest workers. Here, the maximum latency of BBS remains low, and only slightly
increases with more tasks. The reason has been discussed in Section \ref{sec:exp_real}. In addition, iCrowd achieves low latency when the number of tasks is lower than 2K but achieve much higher latency than fGreedy, BBS and RBS when $m$ increases to 3K and above.

\begin{figure}[t!]
	\centering
	\subfigure[][{ Maximum Latency}]{
		\scalebox{0.18}[0.18]{\includegraphics{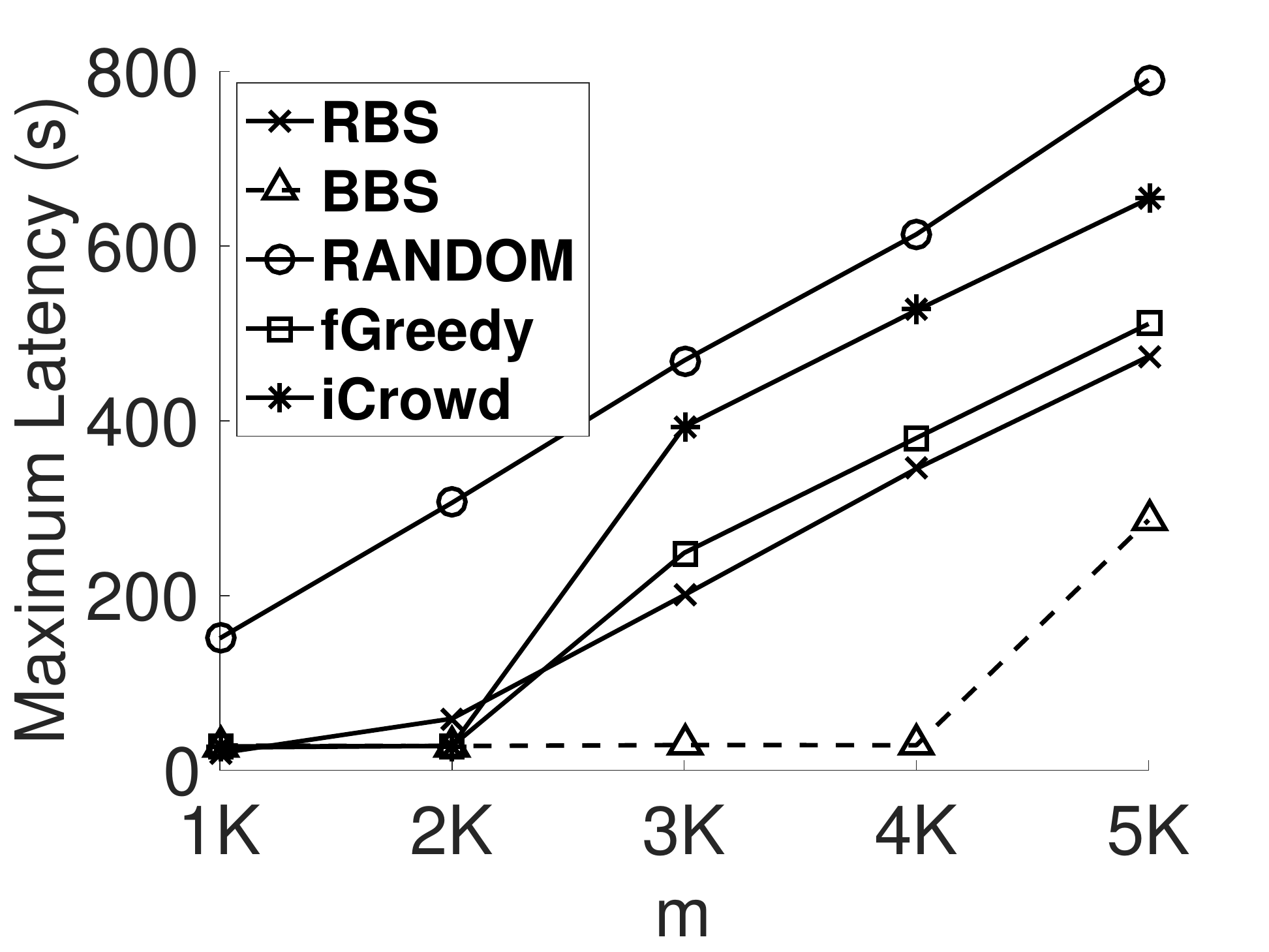}}
		\label{subfig:task_latency}}
	\subfigure[][{ Average Accuracy}]{
		\scalebox{0.18}[0.18]{\includegraphics{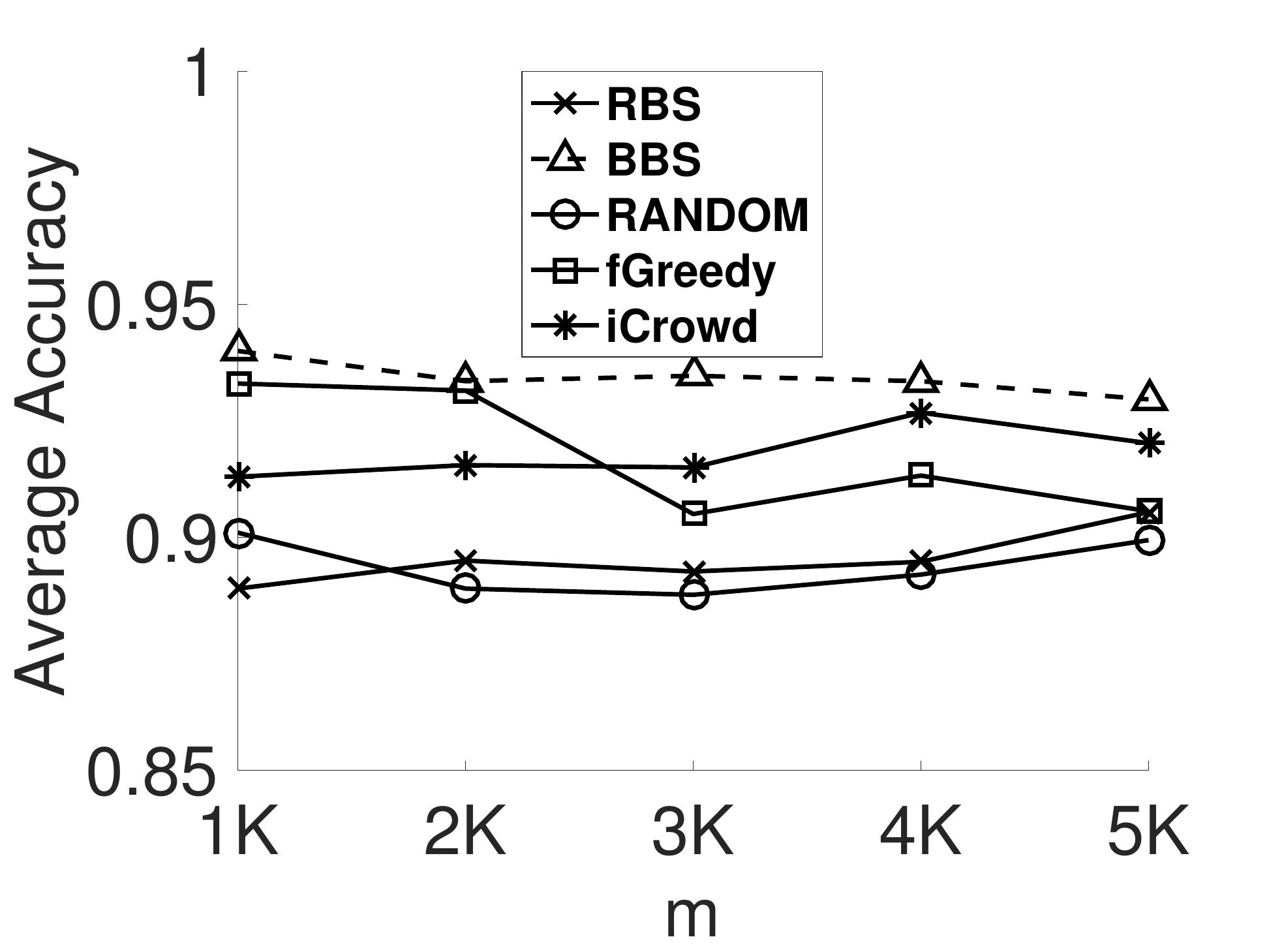}}
		\label{subfig:task_accuarcy}}
	\caption{ Effect of the number of tasks $m$.}
	\label{fig:task_fig}
\end{figure}

\begin{figure}[t!]
	\centering
	\subfigure[][{ Maximum Latency}]{
		\scalebox{0.18}[0.18]{\includegraphics{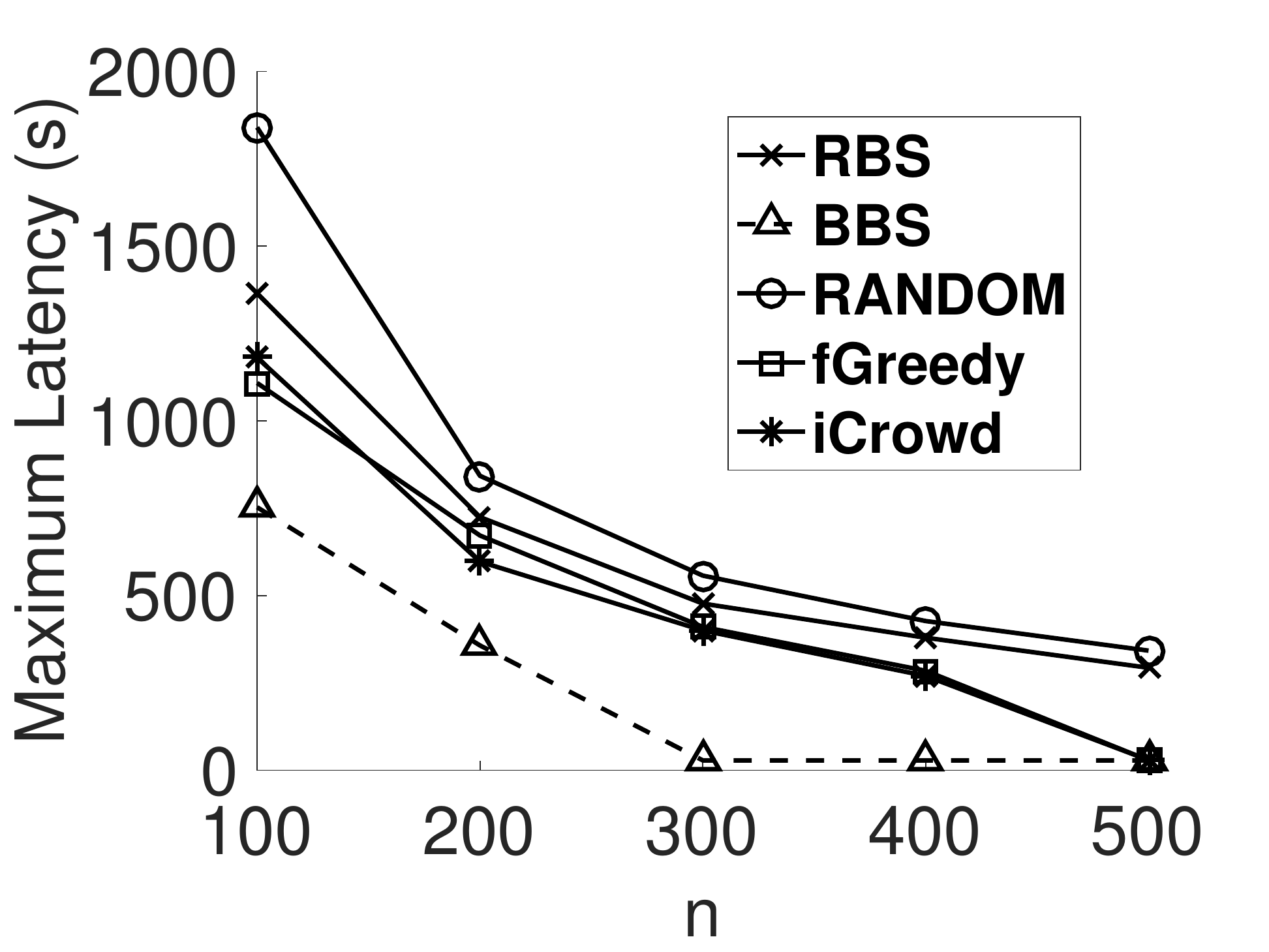}}
		\label{subfig:worker_latency}}
	\subfigure[][{ Average Accuracy}]{
		\scalebox{0.18}[0.18]{\includegraphics{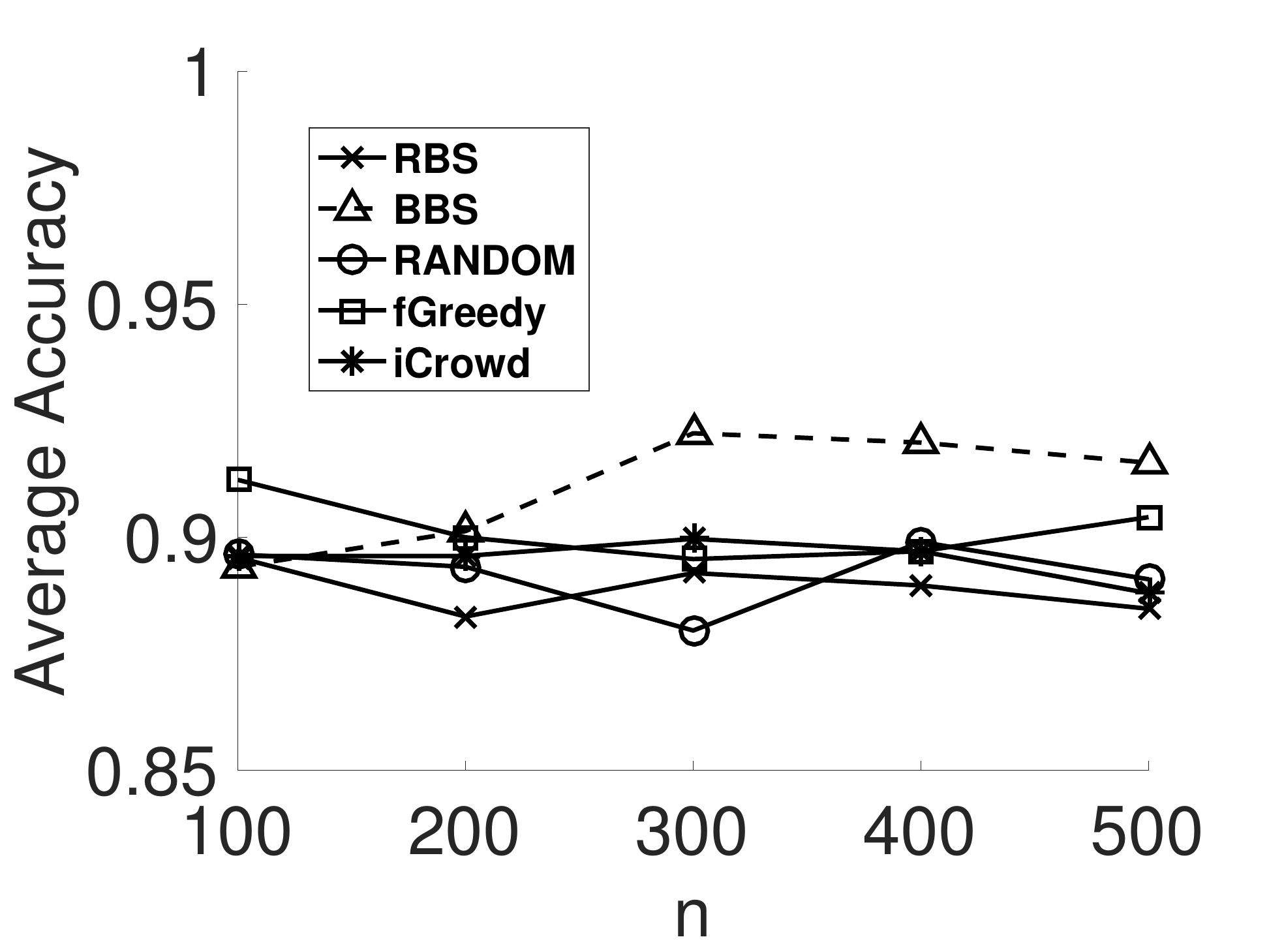}}
		\label{subfig:worker_accuracy}}
	\caption{ Effect of the number of workers $n$.}
	\label{fig:worker_fig}
\end{figure}

Figure \ref{subfig:task_accuarcy} illustrates the average accuracies
of five approaches, with different $m$ values. Since BBS always
chooses a minimum set of workers with the highest category
accuracies, in most cases, the task accuracies of BBS are higher
than the other three approaches. fGreedy can achieve slightly higher accuracy than RBS, as fGreedy can select a set of workers that meets the required accuracy threshold of the task with the highest delay probability while RBS can only determine to assign the current available worker to a suitable task. Nonetheless,
from the figure, RBS and BBS approaches can achieve high task
accuracies (i.e., $89\% \sim 94\%$). We also conducted the experiments with iCrowd having different parameter $k$ on varying number of tasks. Due to space limitation, please refer to \textbf{Appendix C} of supplementary materials for the details.

\noindent \textbf{Effect of the Number, $n$, of Workers.} Figure
\ref{fig:worker_fig} shows the experimental results, where the
number, $n$, of workers changes from 100 to 500, and other
parameters are set to their default values. For the maximum
latencies shown in Figure \ref{subfig:worker_latency}, when the
number, $n$, of worker increases, the maximum latencies of five
algorithms decrease. This is because, with more workers, each task
can be assigned with more workers (potentially with lower
latencies). Since the quality thresholds of tasks are not changing,
with more available workers, the maximum latencies thus decrease.
Similarly, BBS can maintain a much lower maximum latency than the
other four algorithms.  For the average accuracies in Figure
\ref{subfig:worker_accuracy}, our BBS algorithm can achieve
high average accuracies (i.e., $87\% \sim 93\%$).

\noindent \textbf{Effect of the Range of the Quality Threshold
	$[q^-, q^+]$.} Figure \ref{fig:quality_fig} shows the performance
of five approaches, where the range, $[q^-, q^+]$,
of quality thresholds, $q_i$, increases from $[0.75, 0.8]$ to
$[0.9, 0.95]$, and other parameters are set to their default
values. Specifically, as depicted in Figure
\ref{subfig:quality_latency}, when the range of the quality
threshold increases, the maximum latencies of BBS, RBS, RANDOM and fGreedy also increase. The reason is that, with higher quality
threshold $q_i$, each task needs more workers to be satisfied (as
shown by Corollary \ref{coro:worker_increaes}). Similarly, BBS can
achieve much lower maximum latencies than that of RBS, fGreedy and RANDOM.
Further, RBS is better than RANDOM but worse than fGreedy, w.r.t. the maximum latency. Since iCrowd always assigns $k$ (=3) workers to each task, the specific quality thresholds do not affect it w.r.t. the maximum latency.

In Figure \ref{subfig:quality_accuracy}, when the range of $q_i$
increases, the average accuracies of BBS, RBS, RANDOM and fGreedy also
increase. This is because, when $q_i$ increases, each task needs more workers to satisfy its quality threshold (as shown by Corollary \ref{coro:worker_increaes}), which makes the average accuracies of tasks increase. Similar to previous results, our two approaches, BBS and RBS, can achieve higher average accuracies than RANDOM. fGreedy can achieve close accuracy to BBS. Similarly, the specific quality thresholds do not affect it w.r.t. the average accuracies.

\begin{figure}[t!]
	\centering
	\subfigure[][{ Maximum Latency}]{
		\scalebox{0.18}[0.18]{\includegraphics{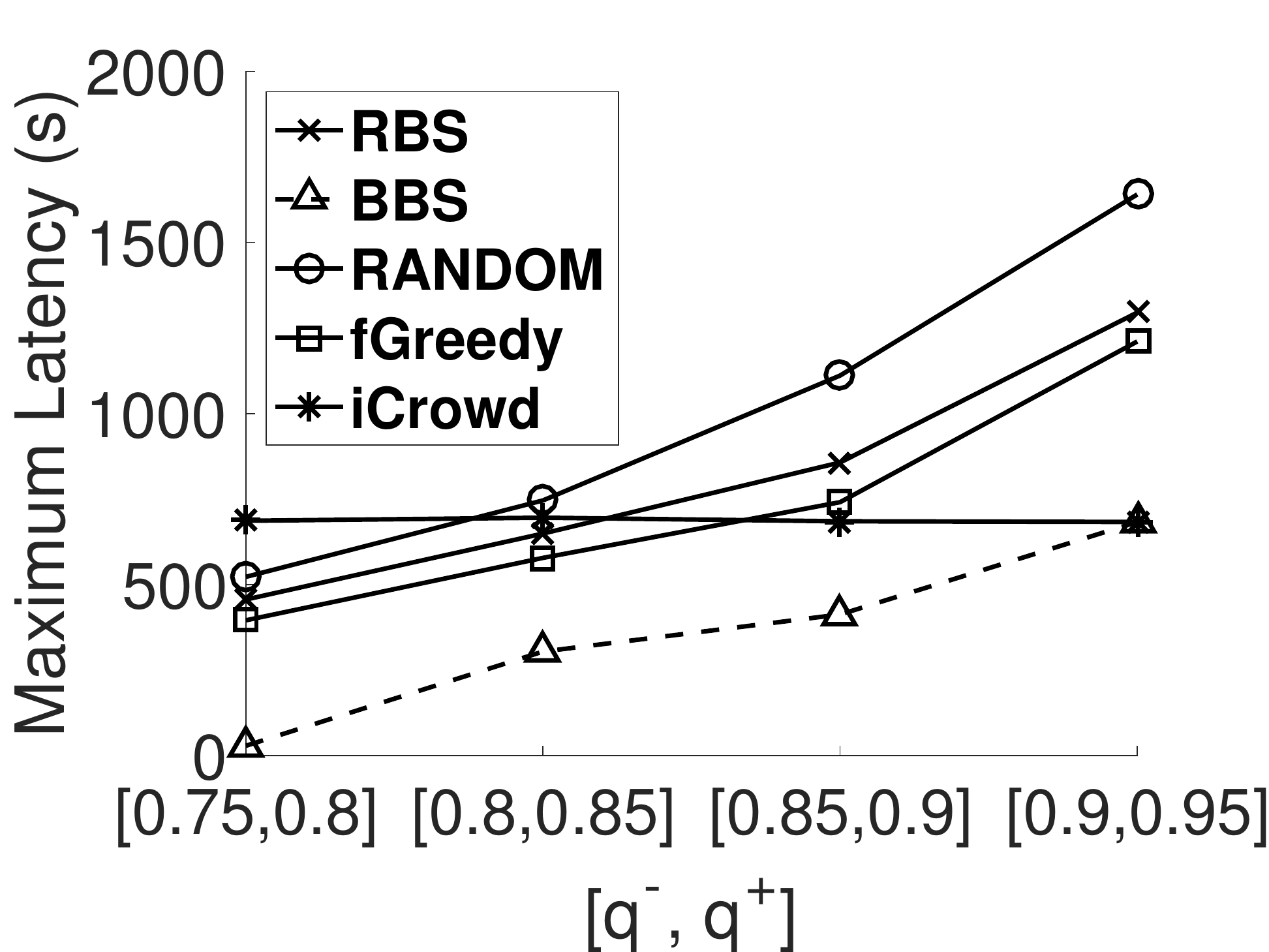}}
		\label{subfig:quality_latency}}
	\subfigure[][{ Average Accuracy}]{
		\scalebox{0.18}[0.18]{\includegraphics{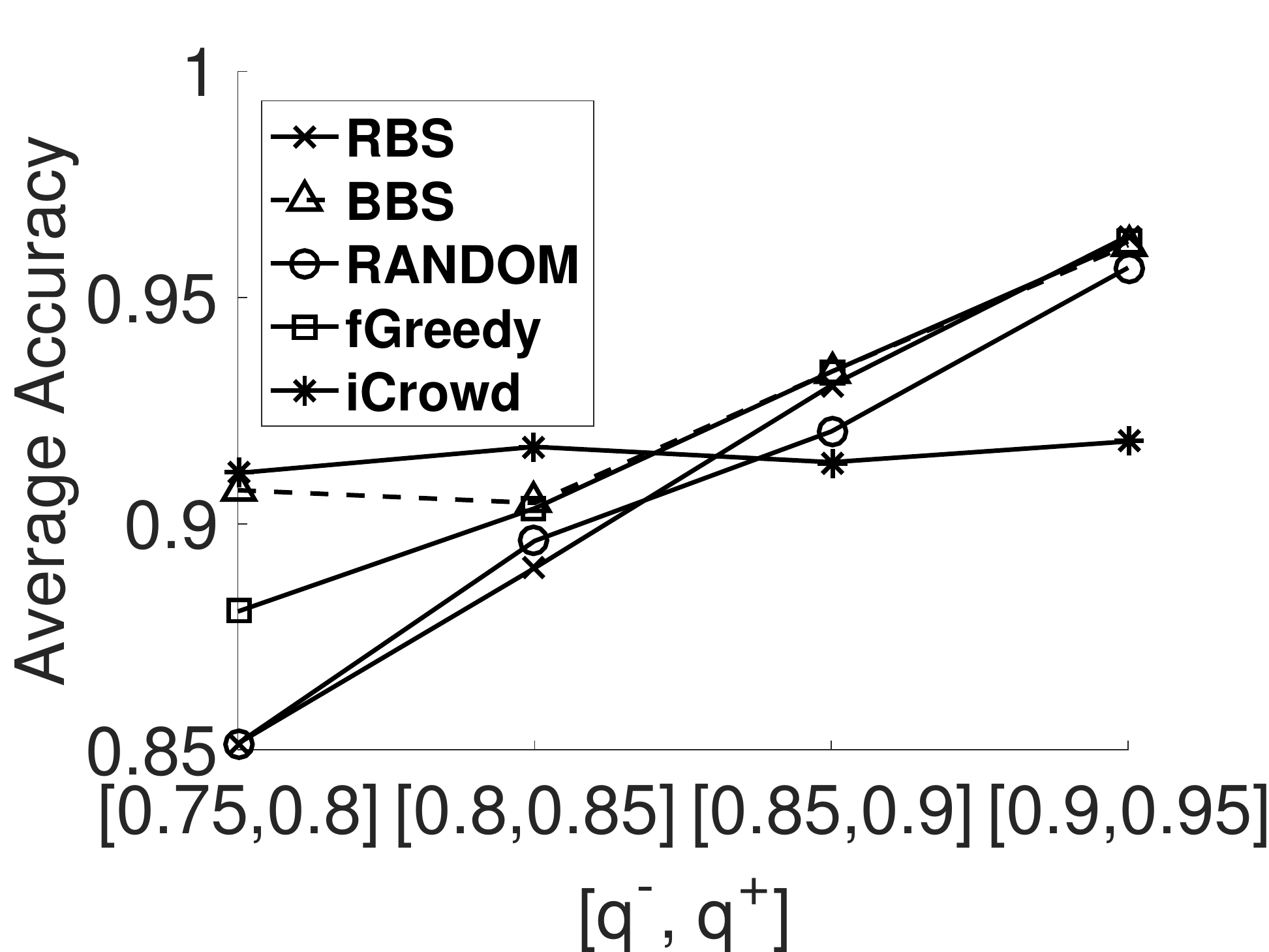}}
		\label{subfig:quality_accuracy}}
	\caption{ {Effect of the specific quality value range $[q^-, q^+]$.}}
	\label{fig:quality_fig}
\end{figure}

\noindent \textbf{Effect of the Number, $l$, of Categories.} 
Figure \ref{fig:category_fig} varies the number, $l$, of categories from 5
to 40, where other parameters are set by default. From Figure
\ref{subfig:c_latency}, we can see that, our RBS and BBS approaches
can both achieve low maximum latencies, with different $l$ values.
Similar to previous results, BBS can achieve the lowest maximum
latencies among three approaches, and RBS is better than RANDOM.
Moreover, in Figure \ref{subfig:c_accuracy}, with different $l$
values, the accuracies of BBS remain high (i.e.,
$92\%\sim94\%$), and are better than that of the other four
algorithms.

\begin{figure}[t!]
	\centering
	\subfigure[][{ Maximum Latency}]{
		\scalebox{0.18}[0.18]{\includegraphics{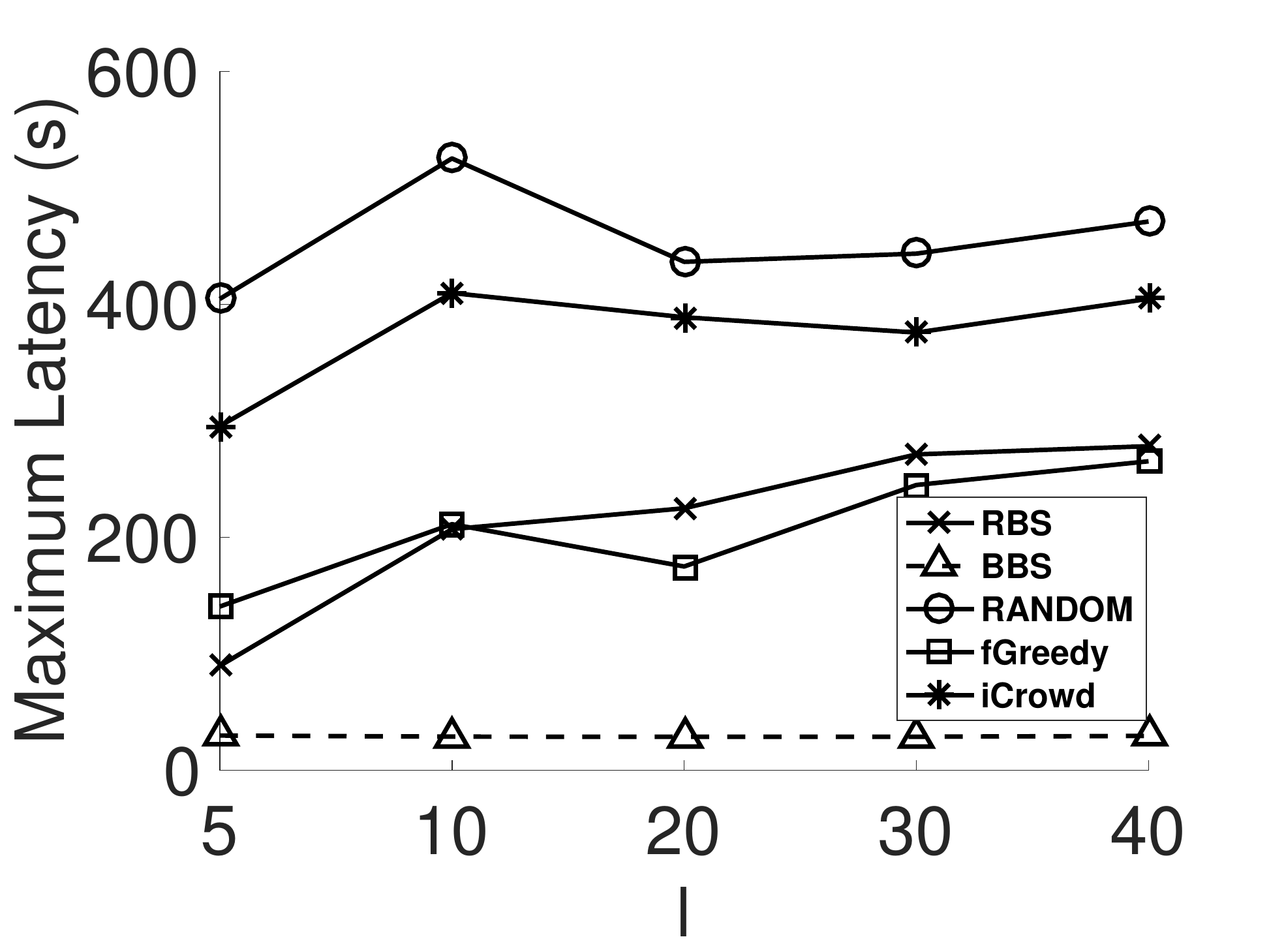}}
		\label{subfig:c_latency}}
	\subfigure[][{ Average Accuracy}]{
		\scalebox{0.18}[0.18]{\includegraphics{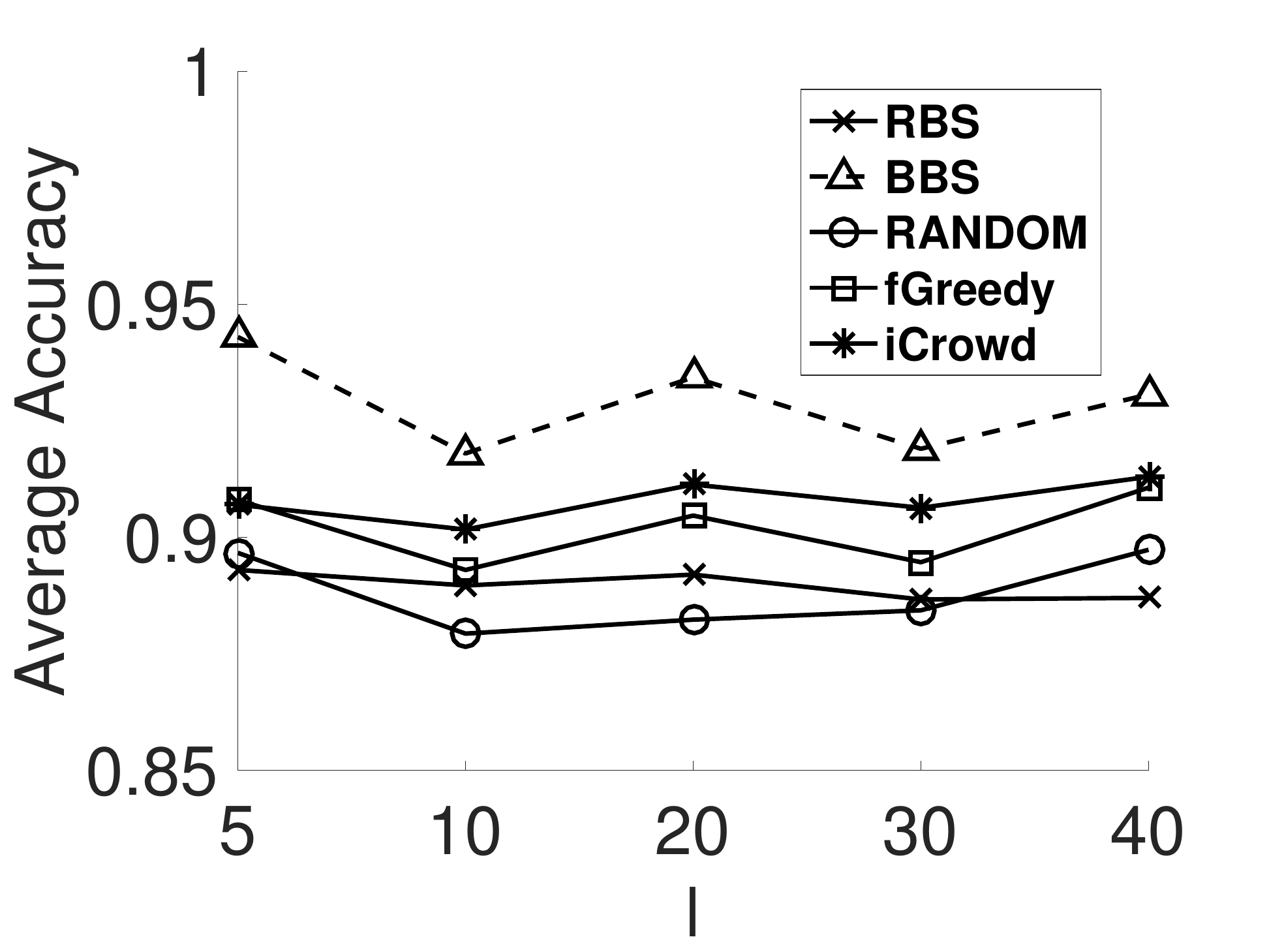}}
		\label{subfig:c_accuracy}}
	\caption{ Effect of the number of categories $l$.}
	\label{fig:category_fig} 
\end{figure}

In summary, our task scheduler module can achieve results with low latencies and high accuracies on both real and synthetic datasets. Especially, our BBS approach is the best one among all the tested scheduling approaches. Moreover, verified through the experiments on the tweet dataset, our smooth KDE model can accurately predict the acceptance probabilities of workers, and achieve higher precision and recall scores than three baseline methods: KDE, Random and NWP.

\section{Related Work}
\label{sec:related}

Crowdsourcing has been well studied by different research
communities (e.g., the database community), and widely used to solve
problems that are challenging for computer (algorithms), but easy
for humans (e.g., sentiment analysis
\cite{mohammad2013crowdsourcing} and entity resolution
\cite{wang2012crowder}). In the databases area, CrowdDB
\cite{franklin2011crowddb} and Qurk \cite{marcus2011crowdsourced}
are designed as crowdsourcing incorporated databases; CDAS
\cite{liu2012cdas} and iCrowd \cite{fan2015icrowd} are systems
proposed to achieve high quality results with crowds; gMission
\cite{chen2014gmission} and MediaQ \cite{kim2014mediaq} are general
spatial crowdsourcing systems that extend crowdsourcing to the real
world. Due to intrinsic error rates of humans,
crowdsourcing systems always focus on achieving high-quality results
with minimum costs. To guarantee the quality of the results, each
task can be answered by multiple workers, and the final result is
aggregated from answers with voting \cite{fan2015icrowd,
	cao2012whom} or learning \cite{liu2012cdas,
	karger2011iterative} methods. To manage the budget, existing studies \cite{karger2014budget, wang2013quality} focused on designing budget-optimal task allocation methods to finish tasks with minimum number of workers with accuracy guarantees and proposed quality-based pricing mechanisms for workers with heterogeneous quality, which do not take the speeds of workers and the latencies of tasks into consideration directly.

Due to the diversity of the workers and their autonomous
participation style in existing crowdsourcing markets (e.g., Amazon
Mechanical Turk (AMT) \cite{amt} and Crowdflower
\cite{crowdflower}), the quality and completion time of
crowdsourcing tasks cannot always be guaranteed. For example, in
AMT, the latency of finishing tasks may vary from minutes to days
\cite{franklin2011crowddb, kittur2008crowdsourcing}. Some difficult
tasks are often ignored by workers, and left uncompleted for a long
time. Recently, several studies \cite{gao2014finish,
	haas2015clamshell, parameswaran2014optimal, verroios2015entity}
focused on reducing the completion time of tasks. In
\cite{parameswaran2014optimal, verroios2015entity}, the authors
designed algorithms to reduce the latencies of tasks for specific
jobs, such as rating and filtering records, and resolve the entities
with crowds. The proposed techniques for specific tasks, however,
cannot be used for general crowdsourcing tasks, which is the target
of our FROG framework. In addition, the cognitive bias of workers may affect the aggregated results, which is considered in the idea selection problem \cite{xu2012reference} by providing the reference ideas. The cognitive bias of workers can be another dimension to handle in our future work. For example, we can model the cognitive bias of workers from their historical answers, then use the cognitive bias to more accurately infer the final results.

Gao et al. \cite{gao2014finish} leveraged the pricing model from
prior studies, and developed algorithms to minimize the total elapsed
time with user-specified monetary constraint or to minimize the
total monetary cost with user-specified deadline constraint. They
utilized the decision theory (specifically, Markov decision
processes) to dynamically modify the prices of tasks. Daniel et al.
\cite{haas2015clamshell} proposed a system, called CLAMShell, to
speed up crowds in order to achieve consistently low-latency data
labeling. They analyzed the sources of labeling latency. To tackle
the sources of latency, they designed several techniques (such as
straggler mitigation to assign the delayed tasks to multiple
workers, and pool maintenance) to improve the average worker speed
and reduce the worker variance of the retainer pool.

Recently,  Goel, Rajpal and Mausam \cite{goel2017octopus} utilize the machine learning methods to study the relationship of the speed, budget and quality of crowdsourcing tasks. They proposed a learning-to-optimizing protocol to simultaneously optimize the budget allocation during each round while minimizing the task latency of a batch of binary tasks (having 0/1 response) and to maximizing the qualities of tasks in open labor markets (e.g., AMT). 

Different from the existing studies \cite{gao2014finish,
	parameswaran2014optimal, verroios2015entity, bernstein2011crowds,
	haas2015clamshell, bernstein2012analytic}, our FROG framework is designed for general
crowdsourcing tasks (rather than specific tasks), and focuses on
both reducing the latencies of all tasks and improving the accuracy
of tasks(instead of either latency or accuracy). In our FROG framework, the task scheduler
module actively assigns workers to tasks with high reliability and
low latency, which takes into account response times and category
accuracies of workers, as well as the difficulties of tasks (not
fully considered in prior studies). We also design two novel
scheduling approaches, request-based and batch-based scheduling.
Different from prior studies \cite{fan2015icrowd, haas2015clamshell} that simply filtered out
workers with low accuracies, our work utilizes all possible worker
labors, by scheduling difficult/urgent tasks to high-accuracy/fast
workers and routing easy and not urgent tasks to low-accuracy
workers.

Moreover, Bernstein et al. \cite{bernstein2011crowds} proposed the
retainer model to hire a group of workers waiting for tasks, such
that the latency of answering crowdsourcing tasks can be
dramatically reduced. Bernstein et al. \cite{bernstein2012analytic}
also theoretically analyzed the optimal size of the retainer model
using queueing theory for realtime crowdsourcing, where
crowdsourcing tasks come individually. These models may either
increase the system budget or encounter the scenario where online
workers are indeed not enough for the assignment during some period.
In contrast, with the help of smart devices, our FROG framework has
the capability to invite offline workers to do tasks, which can
enlarge the public worker pool, and enhance the throughput of the
system. In particular, our notification module in FROG can contact
workers who are not online via smart devices, and intelligently send
invitation messages only to those available workers with high
probabilities. Therefore, with the new model and different goals in
our FROG framework, we cannot directly apply techniques in previous
studies to tackle our problems (e.g., FROG-TS and EWN).

\section{Conclusion}
\label{sec:conclusion}

The crowdsourcing has played an important role in many real
applications that require the intelligence of human workers (and
cannot be accurately accomplished by computers or algorithms), which
has attracted much attention from both academia and industry. In
this paper, inspired by the accuracy and latency problems of
existing crowdsourcing systems, we propose a novel \textit{fast and
	reliable crowdsourcing} (FROG) framework, which actively assigns
workers to tasks with the expected high accuracy and low latency
(rather than waiting for autonomous unreliable and high-latency
workers to select tasks). We formalize the FROG task scheduling
(FROG-TS) and efficient worker notifying (EWN) problems, and
proposed effective and efficient approaches (e.g., request-based,
batch-based scheduling, and smooth KDE) to enable the FROG
framework. Through extensive experiments, we demonstrate the
effectiveness and efficiency of our proposed FROG framework on both
real and synthetic data sets.

\section{Acknowledgment}
\label{sec:ack}

Peng Cheng, Xun Jian and Lei Chen are supported in part by the Hong Kong RGC Project
16207617, Science and Technology Planning Project of Guangdong
Province, China, No. 2015B010110006, NSFC Grant No.
61729201, 61232018, Microsoft Research Asia Collaborative Grant
and NSFC Guang Dong Grant No. U1301253. Xiang Lian is supported by Lian Start Up No. 220981, Kent State University.

\bibliographystyle{abbrv}
\bibliography{all}

\input{SupplementaryMaterials.tex}

\vspace{-4ex}
\begin{IEEEbiography}[{\includegraphics[width=1in,height=1.25in,clip,keepaspectratio]{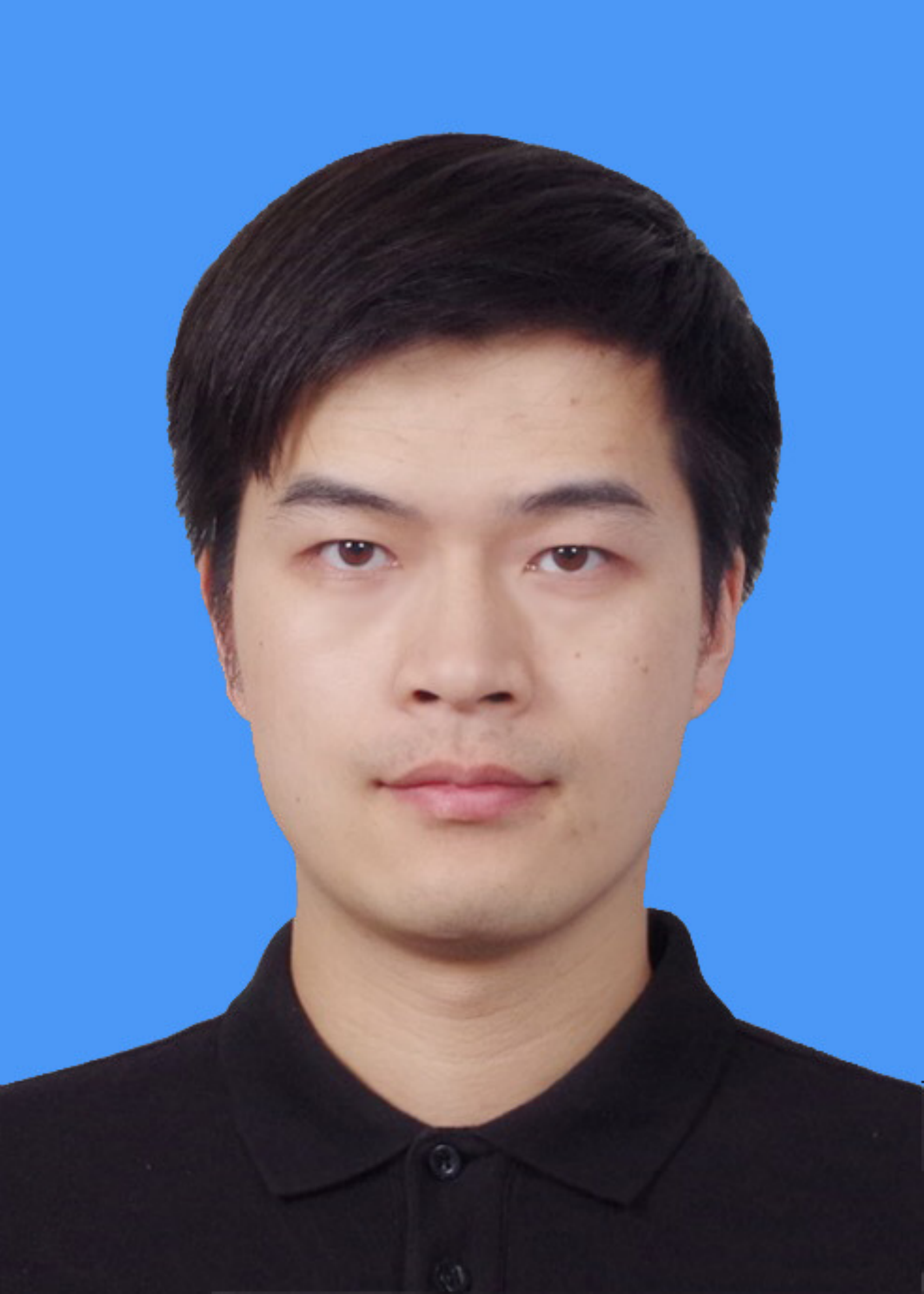}}]
	{Peng Cheng} received his BS degree and MA degree in Software
	Engineering in 2012 and 2014, from Xi'an Jiaotong University, China, and the PhD degree in Computer Science and Engineering from the Hong Kong University of Science and Technology  (HKUST) in 2017.
	He is now a Post-Doctoral Fellow in the Department of Computer Science and
	Engineering in HKUST. His
	research interests include crowdsourcing, spatial crowdsourcing and ridesharing.
\end{IEEEbiography}\vspace{-4ex}

\begin{IEEEbiography}[{\includegraphics[width=1in,height=1.25in,clip,keepaspectratio]{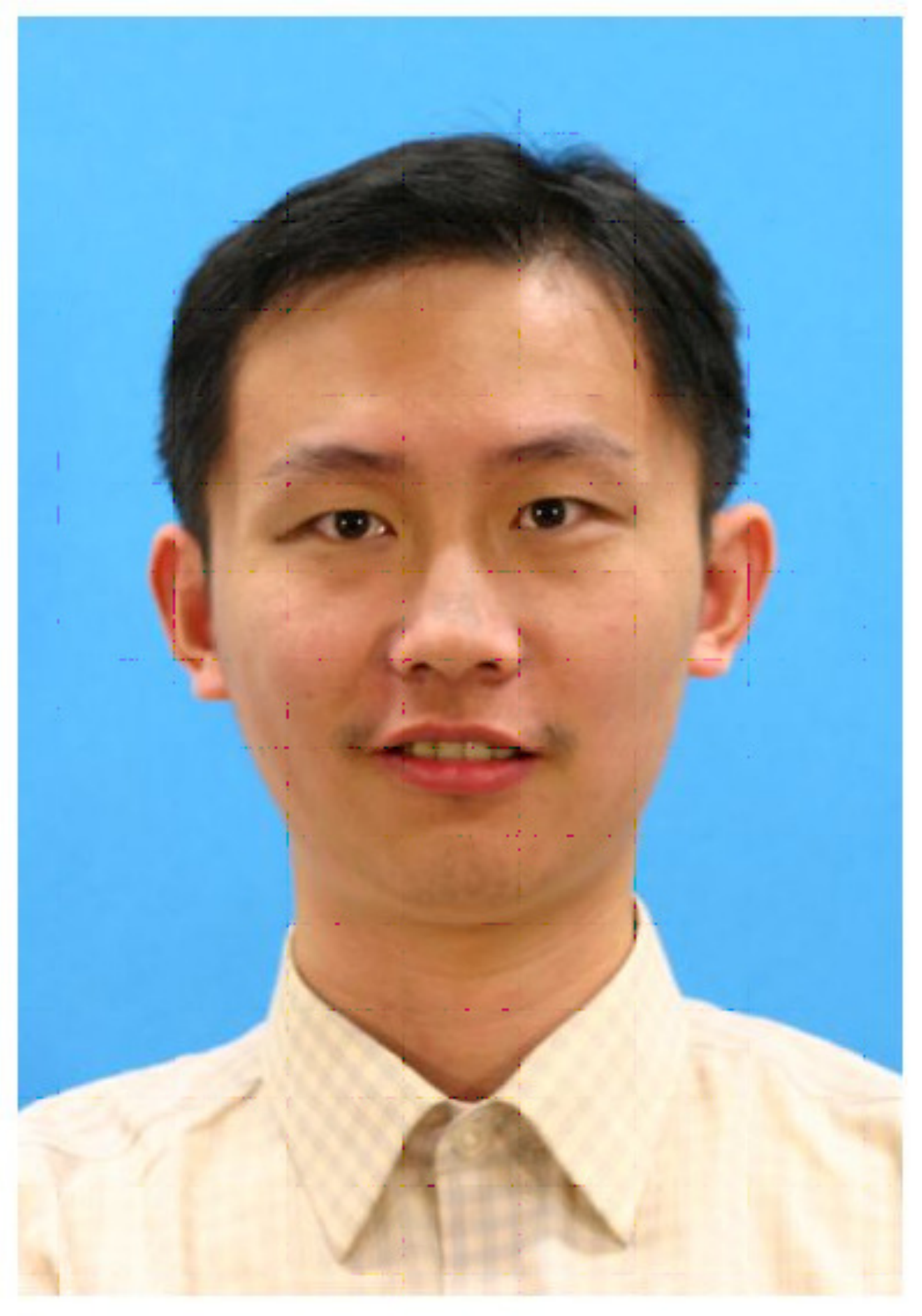}}]
	{Xiang Lian} received the BS degree from the Department of Computer
	Science and Technology, Nanjing University, and the PhD degree in
	computer science from the Hong Kong University of Science and
	Technology. He is now an assistant professor in the Computer Science
	Department at the Kent State University. His
	research interests include probabilistic/uncertain/inconsistent,
	uncertain/certain graph, time-series, and spatial databases.
\end{IEEEbiography}\vspace{-4ex}
\begin{IEEEbiography}[{\includegraphics[width=1in,height=1.25in,clip,keepaspectratio]{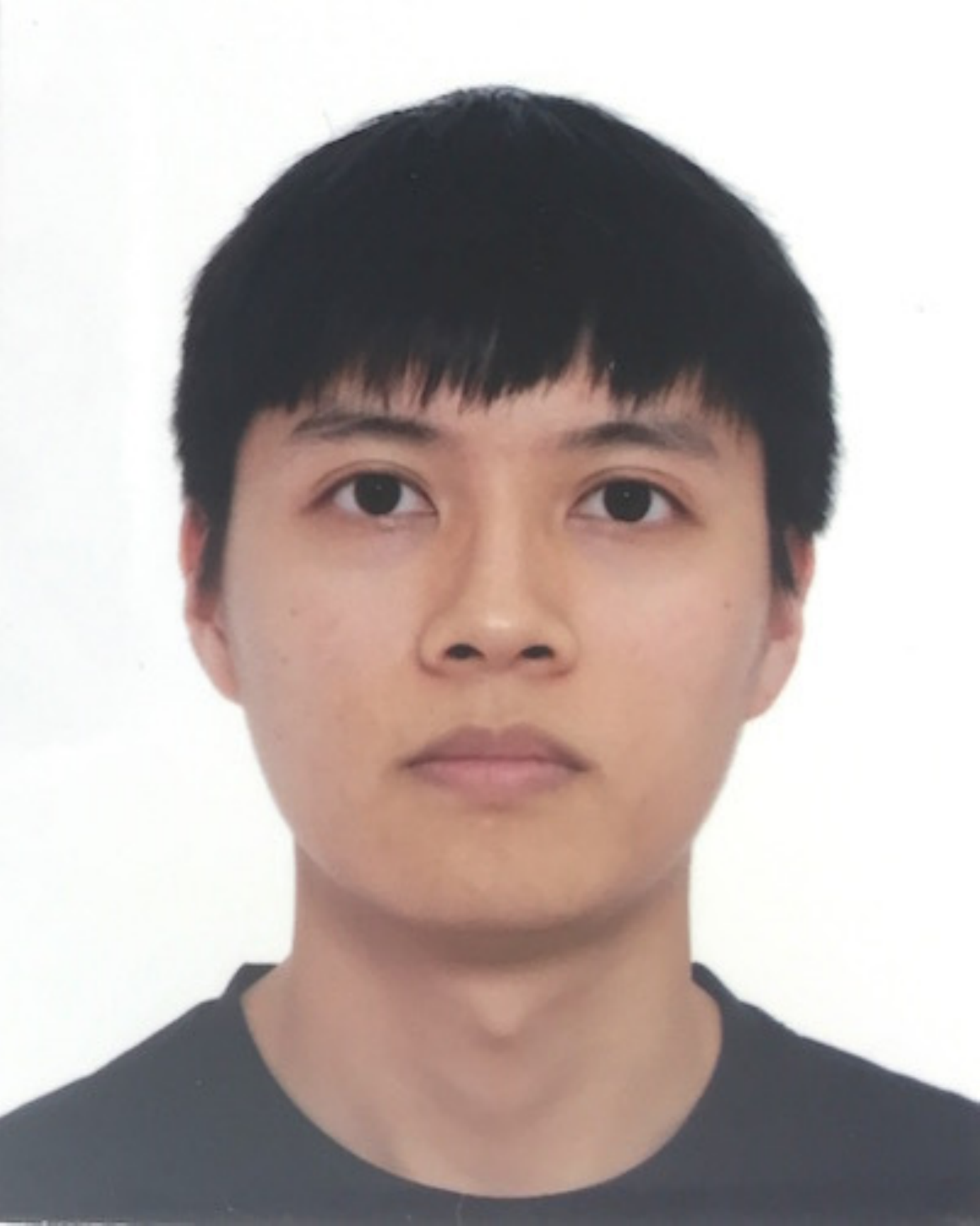}}]
	{Xun Jian} received his B.Eng. degree in Software Engineering in 2014 from Beihang University. Then he received his M.Sc. degree in Information Technology in 2016 from The Hong Kong University of Science and Technology(HKUST). Now he is a Ph.D. student in the Department of Computer Science at HKUST. His research interests include crowdsourcing and algorithms on graph.
\end{IEEEbiography}
\begin{IEEEbiography}[{\includegraphics[width=1in,height=1.25in,clip,keepaspectratio]{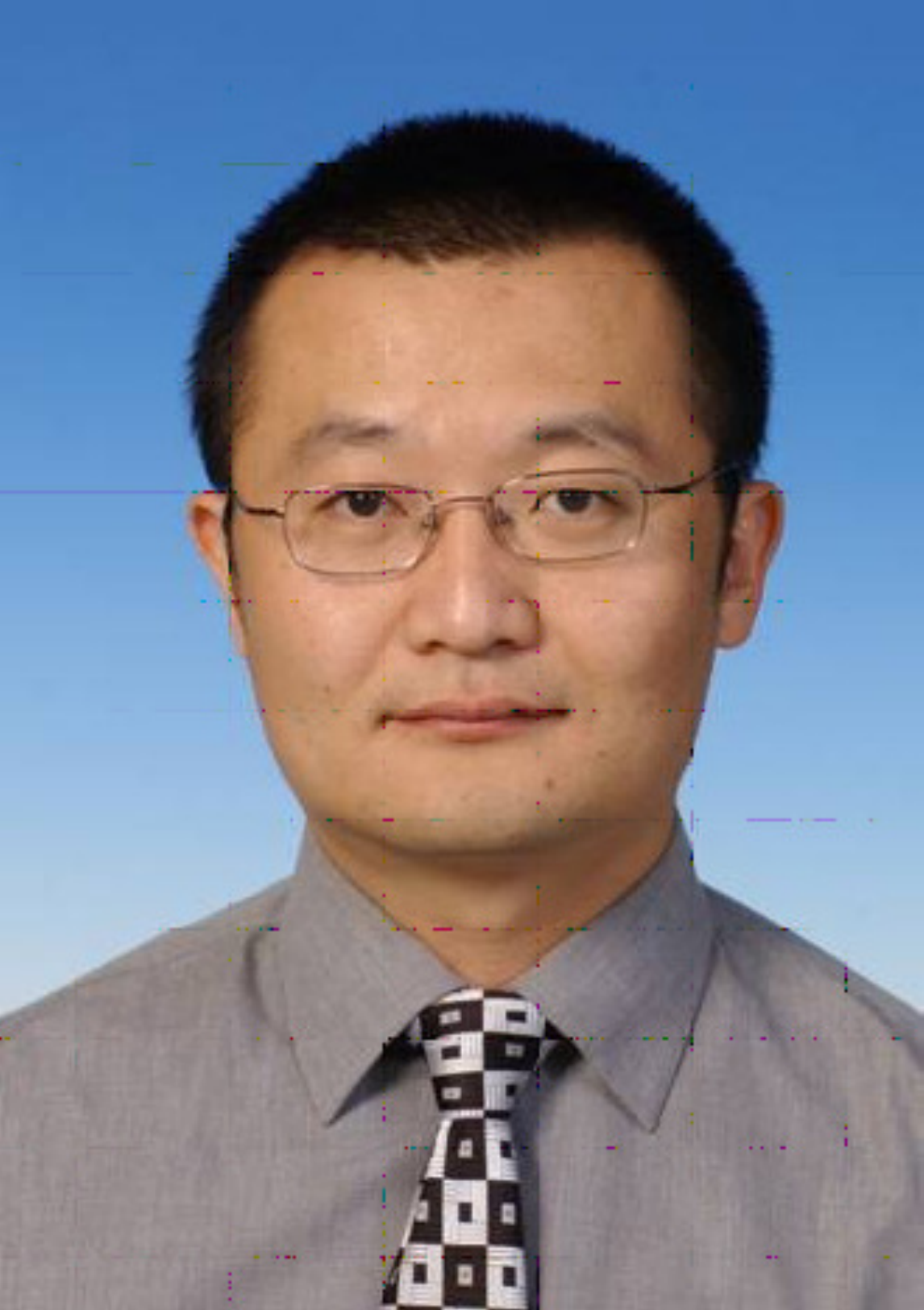}}]
	{Lei Chen} received his BS degree in Computer Science and
	Engineering from Tianjin University, China, in 1994, the MA degree
	from Asian Institute of Technology, Thailand, in 1997, and the PhD
	degree in computer science from University of Waterloo, Canada, in
	2005. He is now a professor in the Department of Computer
	Science and Engineering at Hong Kong University of Science and
	Technology. His research interests include crowdsourcing, uncertain
	and probabilistic databases, multimedia and
	time series databases, and privacy. He is a member of the IEEE.
\end{IEEEbiography}

\end{document}

%% file: SupplementaryMaterials.tex
\section*{Supplementary Materials}

\textbf{Appendix A. Expected Accuracy of Multi-choices Task}
\label{apd:expected_accuracy}

\noindent \textbf{Majority voting with multiple choices.}
Given a task $t_i$ in category $c_l$ and a set of $k$ workers $W_i$ assigned to it, when we use majority voting with $R$ choices, the expected accuracy of tasks $t_i$ can be calculated as follows:
	\begin{align}
	&\Pr(W_i, c_l) \notag\\&= \sum_{x = \lceil\frac{k}{R} \rceil, |W_{i,x}| \text{is max}}^{k}\sum_{W_{i,x}}\Big(\prod_{w_j \in W_{i,x}}\alpha_{jl}\prod_{w_j \in W_i - W_{i,x}}(1-\alpha_{jl})\Big), \notag
	\end{align}
\noindent where $W_{i,x}$ is a subset of $W_i$ with $x$ elements.

\noindent \textbf{Weighted Majority voting with multiple choices.}
Given a task $t_i$ in category $c_l$ and a set of $k$ workers $W_i$ assigned to it, when we use weighted majority voting with $R$ choices, the expected accuracy of tasks $t_i$ can be calculated as follows:
	\begin{align}
	&\Pr(W_i, c_l) \notag\\&= 
	\sum_{x = \lceil\frac{k}{R} \rceil, We(W_{i,x}) \text{is max}}^{k}\sum_{W_{i,x}}\Big(\prod_{w_j \in W_{i,x}}\alpha_{jl} \prod_{w_j \in W_i - W_{i,x}}(1-\alpha_{jl})\Big), \notag
	\end{align}
\noindent where $W_{i,x}$ is a subset of $W_i$ with $x$ elements, and $We(W)$ is the weight of a given worker set $W$.

\noindent \textbf{Half voting with multiple choices.} Half voting only return the results selected by more than half workers.
Given a task $t_i$ in category $c_l$ and a set of $k$ workers $W_i$ assigned to it, when we use half voting with $R$ choices, the expected accuracy of tasks $t_i$ can be calculated as follows:
	\begin{equation}
	\Pr(W_i, c_l) = \sum_{x = \lceil\frac{k}{2} \rceil}^{k}\sum_{W_{i,x}}\Big(\prod_{w_j \in W_{i,x}}\alpha_{jl}\prod_{w_j \in W_i - W_{i,x}}(1-\alpha_{jl})\Big), \notag
	\end{equation}
\noindent where $W_{i,x}$ is a subset of $W_i$ with $x$ elements. Half voting is effective when there are more than two choices and the expected accuracy of each task is calculated by same equation same with that of majority voting with two choices.

\noindent \textbf{Bayesian voting with multiple choices.} Bayesian voting returns the results based on the priors of the choices and the accuracy of the workers. 
Given a task $t_i$ in category $c_l$ and a set of $k$ workers $W_i$ assigned to it, when we use Bayesian voting with $R$ choices, the expected accuracy of tasks $t_i$ can be calculated as follows:
	\begin{align}
	&\Pr(W_i, c_l)\notag\\ &= 
	\sum_{BP(W^r_{i}) \text{is max}}\sum_{W^r_{i}}\Big(BP(W^r_{i})\prod_{w_j \in W_i - W^r_{i}}(1-\Pr(r))(1-\alpha_{jl})\Big), \notag
	\end{align}
\noindent where $W^r_{i}$ is a subset of $W_i$ who select the $r$-th choice, $\Pr(r)$ is the prior probability that the true answer is the $r$-th choice, and $BP(W^r_{i}) = \prod_{w_j \in W^r_i} Pr(r)\alpha_{jl}$ is the Bayesian probability of the $r$-th choice. In Bayesian voting, the returned answer is the $r$-th choice if $BP(W^r_{i})$ is the maximum value among $R$ choices.

\noindent\textbf{Appendix B. Quality of the Results Estimated with the Expectation Maximization Method.}

\begin{figure}[h!]
	\centering
	\subfigure[][{\scriptsize Maximum Latency}]{
		\scalebox{0.18}[0.18]{\includegraphics{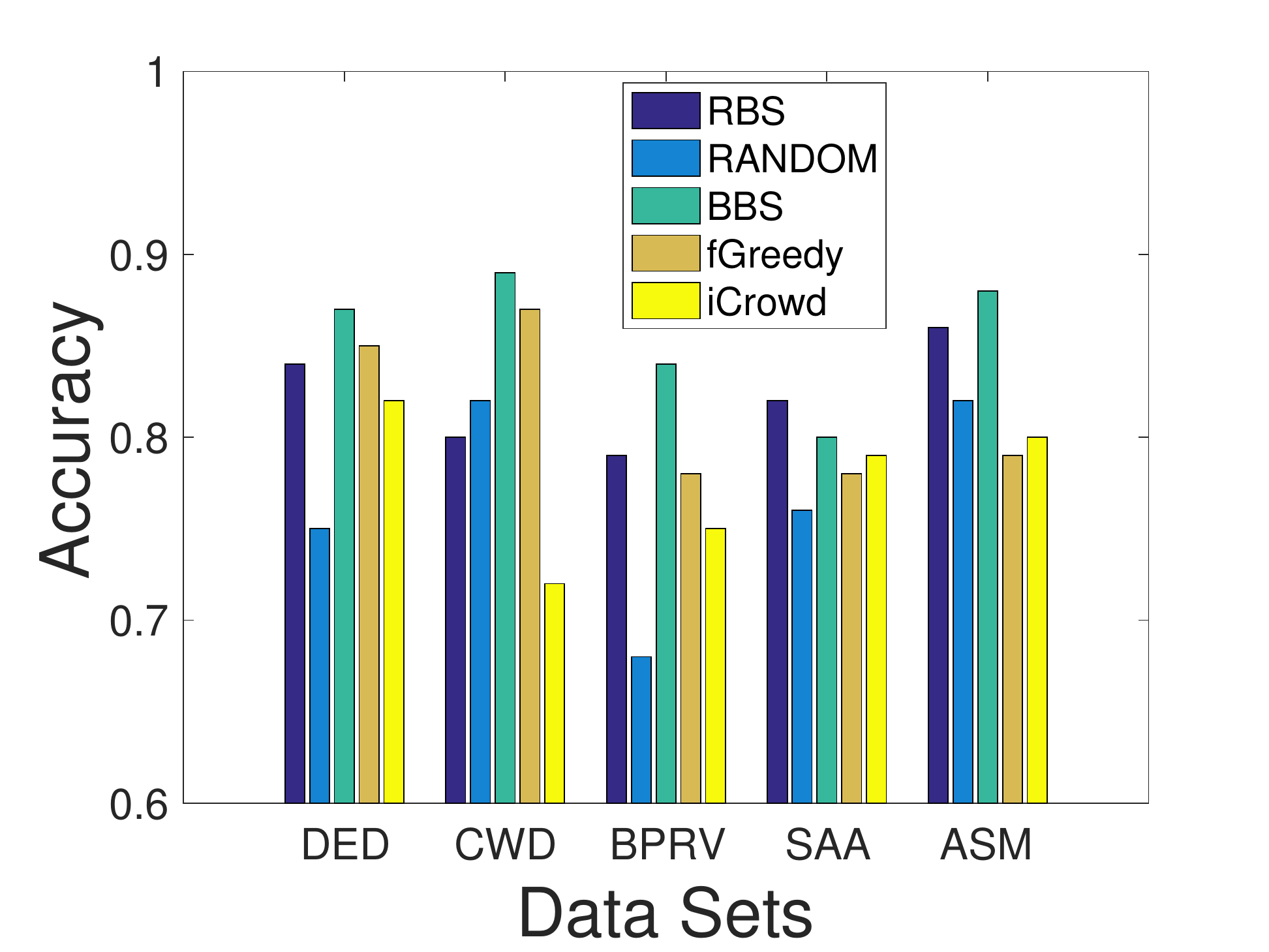}}
		\label{subfig:real_em}}
	\subfigure[][{\scriptsize Average Accuracy}]{
		\scalebox{0.18}[0.18]{\includegraphics{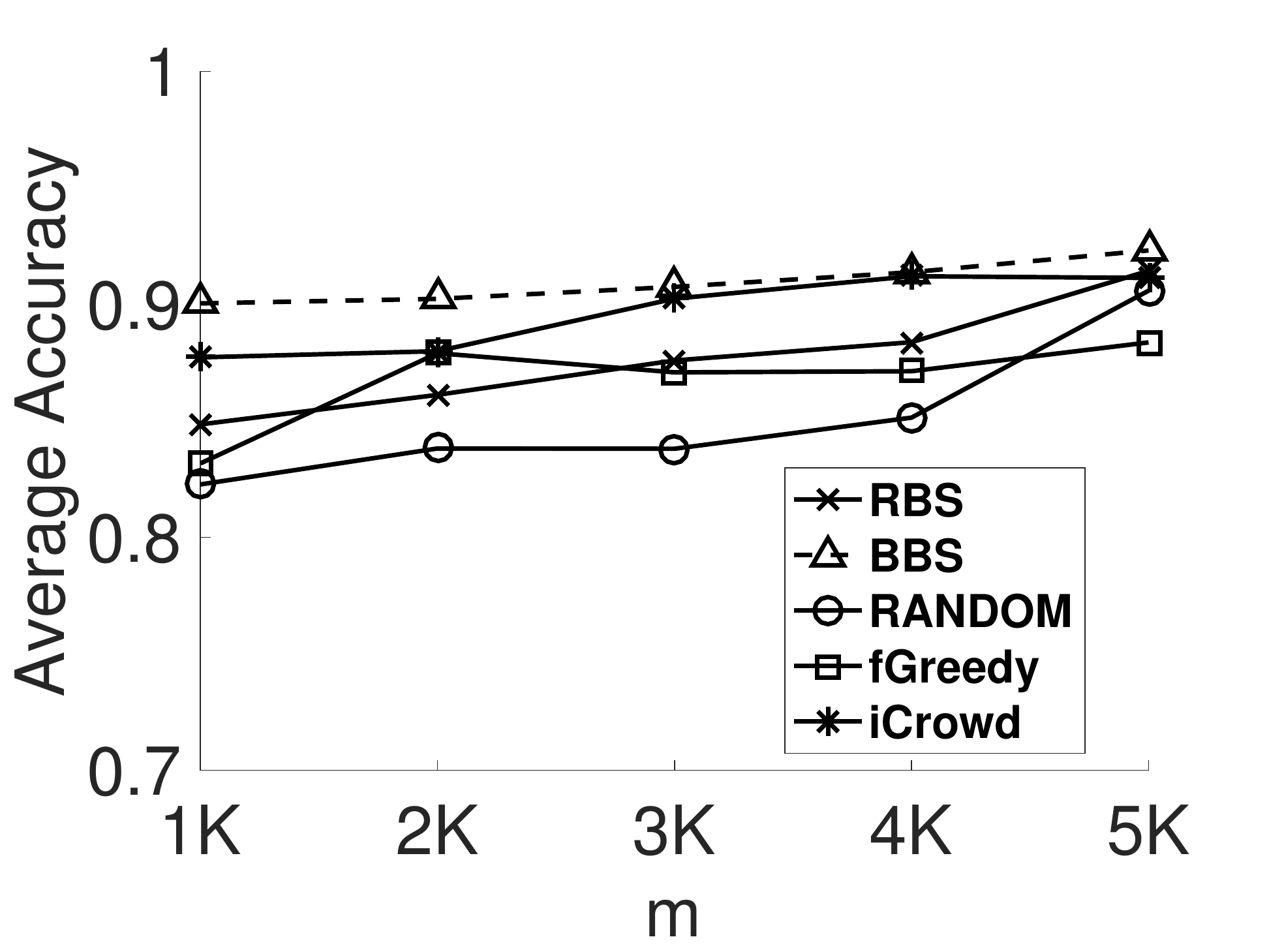}}
		\label{subfig:task_em}}
	\subfigure[][{\scriptsize Average Accuracy}]{
		\scalebox{0.18}[0.18]{\includegraphics{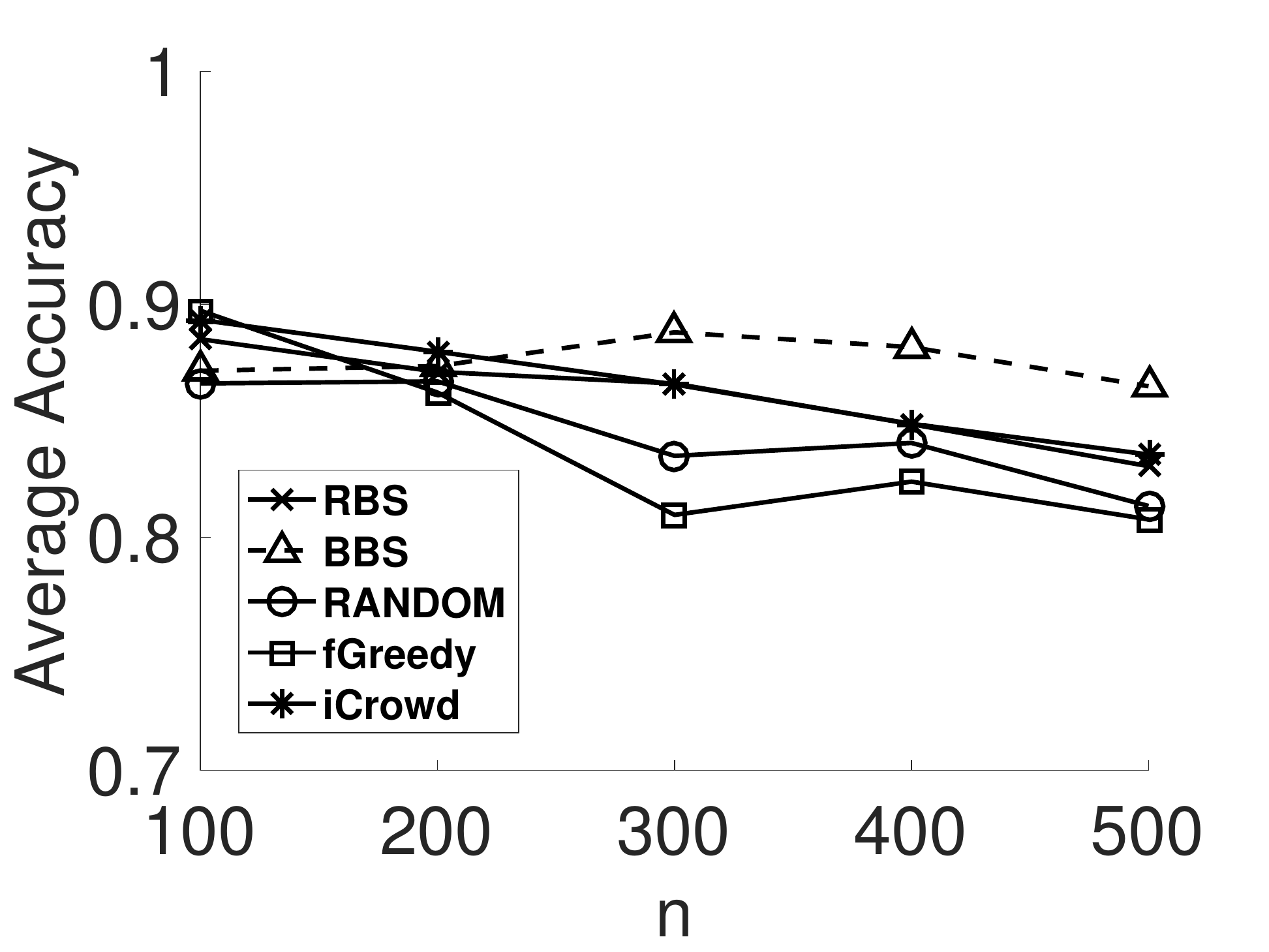}}
		\label{subfig:worker_em}}
	\subfigure[][{\scriptsize Average Accuracy}]{
		\scalebox{0.18}[0.18]{\includegraphics{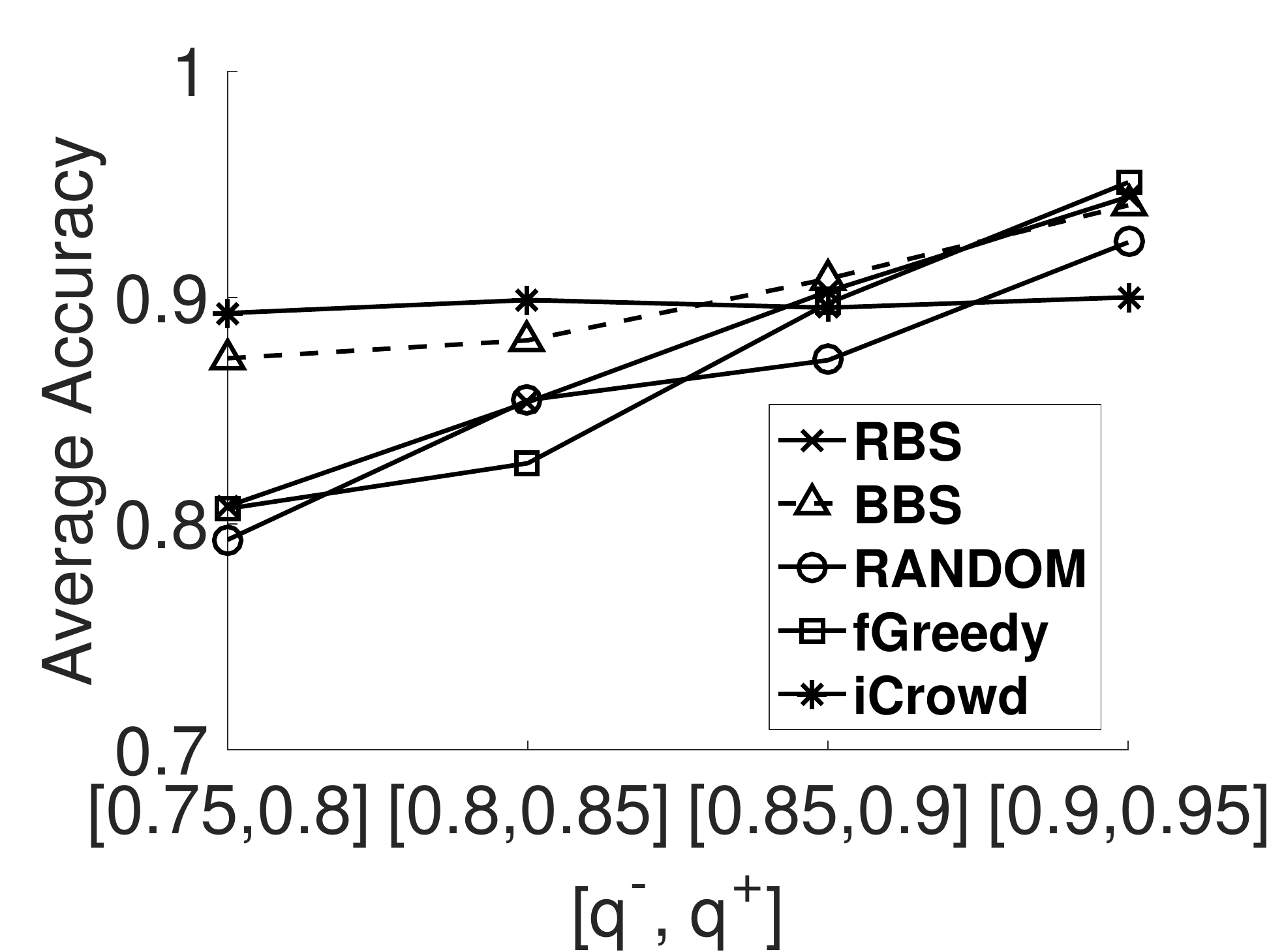}}
		\label{subfig:quality_em}}
	\subfigure[][{\scriptsize Average Accuracy}]{
		\scalebox{0.18}[0.18]{\includegraphics{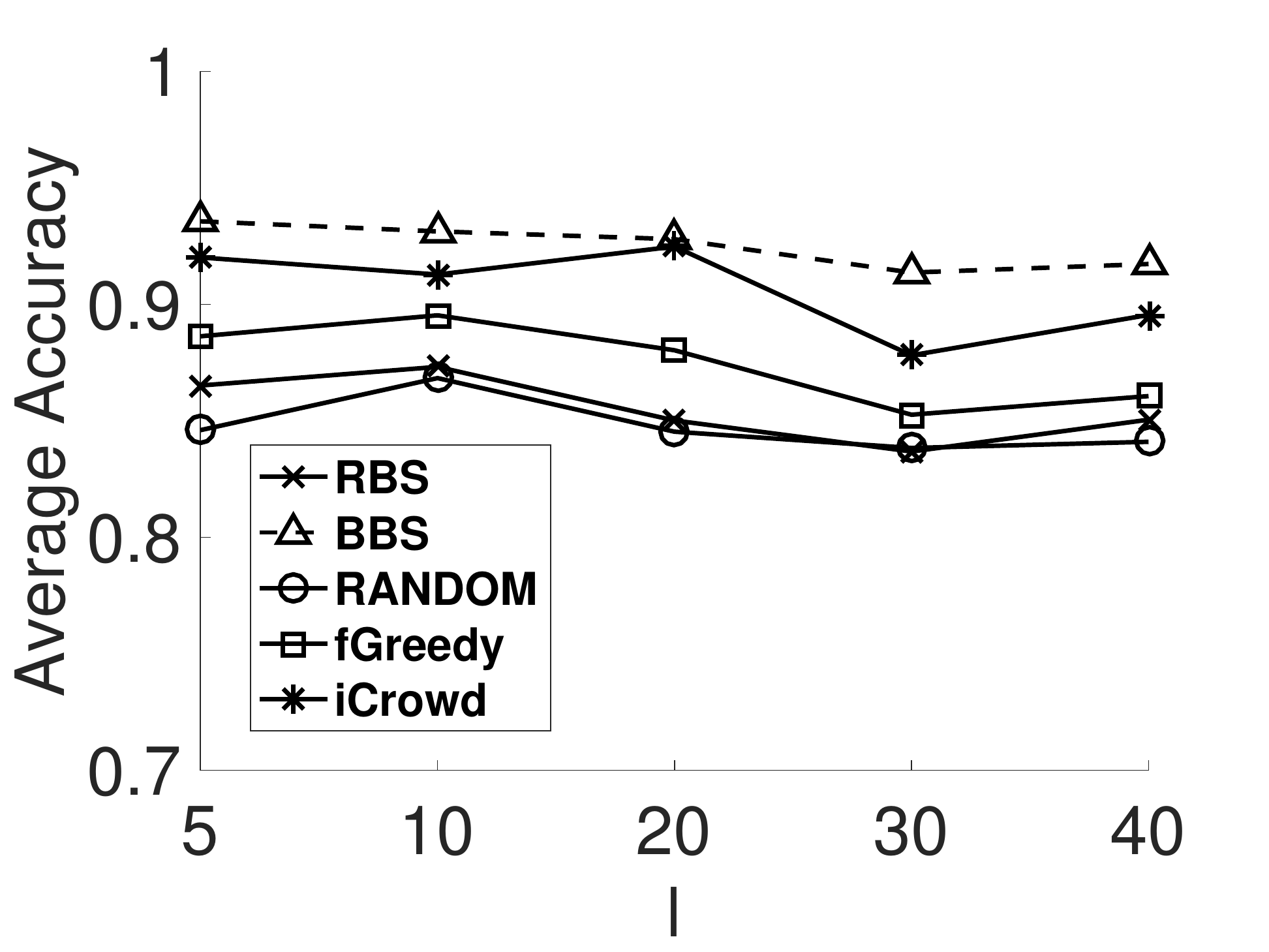}}
		\label{subfig:category_em}}
	\caption{\small Quality of the Results Estimated with the Expectation Maximization Method.}
	\label{fig:em_quality}
\end{figure}

We utilize the expectation maximization algorithm proposed by Dawid and Skene [15], using maximum likelihood, for inferring the error rates of workers who contribute to the micro-tasks and the accuracies of the tasks. Ipeirotis et al. [24]  implements Dawid and Skene's algorithm and open source it (source code: \url{https://github.com/ipeirotis/Get-Another-Label}), which is used to estimate the average accuracies in our experiments about the task scheduler module. In Figure \ref{fig:em_quality}, we present the average accuracies of tested algorithm in the experiments on both real and synthetic data sets.

Similar to the results estimated with the majority voting method on the real data sets, in Figure \ref{subfig:real_em}, our BBS can achieve the highest accuracies on DED, CWD, BPRV and ASM. On SAA, RBS can achieve the highest accuracy. iCrowd achieves higher accuracies than fGreedy on SAA and ASM. 

For the experiments on synthetic data sets, in Figures \ref{subfig:task_em}, \ref{subfig:worker_em} and \ref{subfig:category_em}, RBS, BBS, iCrowd and fGreedy can achieve high and close average accuracies. Specifically, our BBS can achieve slightly higher average accuracies than other tested algorithm. In Figure \ref{subfig:quality_em}, when the required expected qualities of workers are lower than [0.85, 0.9], iCrowd can achieve the highest average accuracies. When the required expected qualities of tasks are in the range of [0.85, 0.9], the tested algorithms except RANDOM can achieve close average accuracies, which are higher than that of RANDOM. However, when the required expected qualities of tasks are in the range of [0.9, 0.95], iCrowd achieves the lowest average accuracy. The reason is that iCrowd constantly assigns $k$ workers to each task while other algorithms dynamically assign workers to tasks to satisfy the required qualities.

\begin{figure}[h!]
	\centering
	\subfigure[][{\scriptsize Maximum Latency}]{
		\scalebox{0.18}[0.18]{\includegraphics{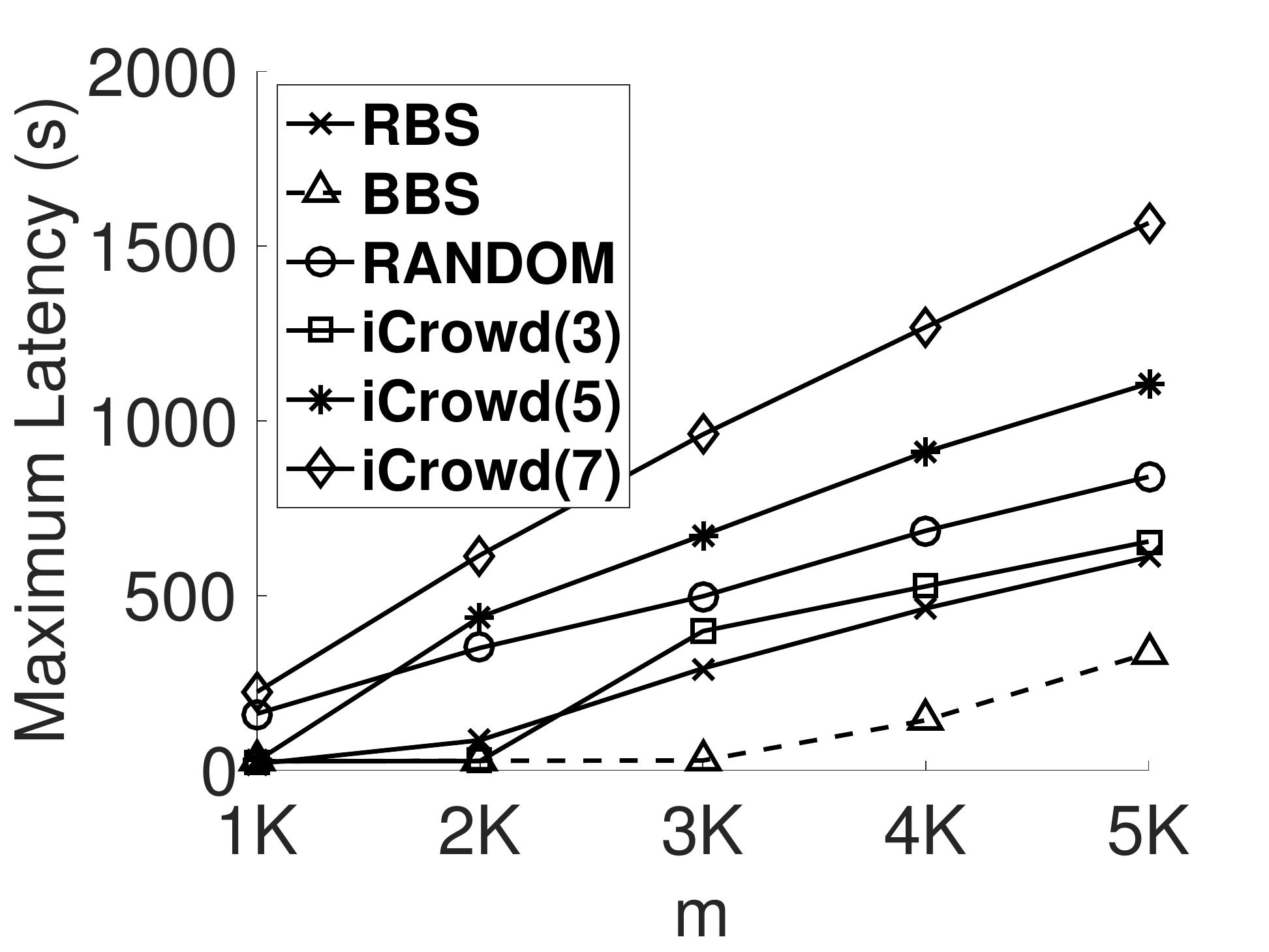}}
		\label{subfig:task_latency_icrowd}}
	\subfigure[][{\scriptsize Average Accuracy}]{
		\scalebox{0.18}[0.18]{\includegraphics{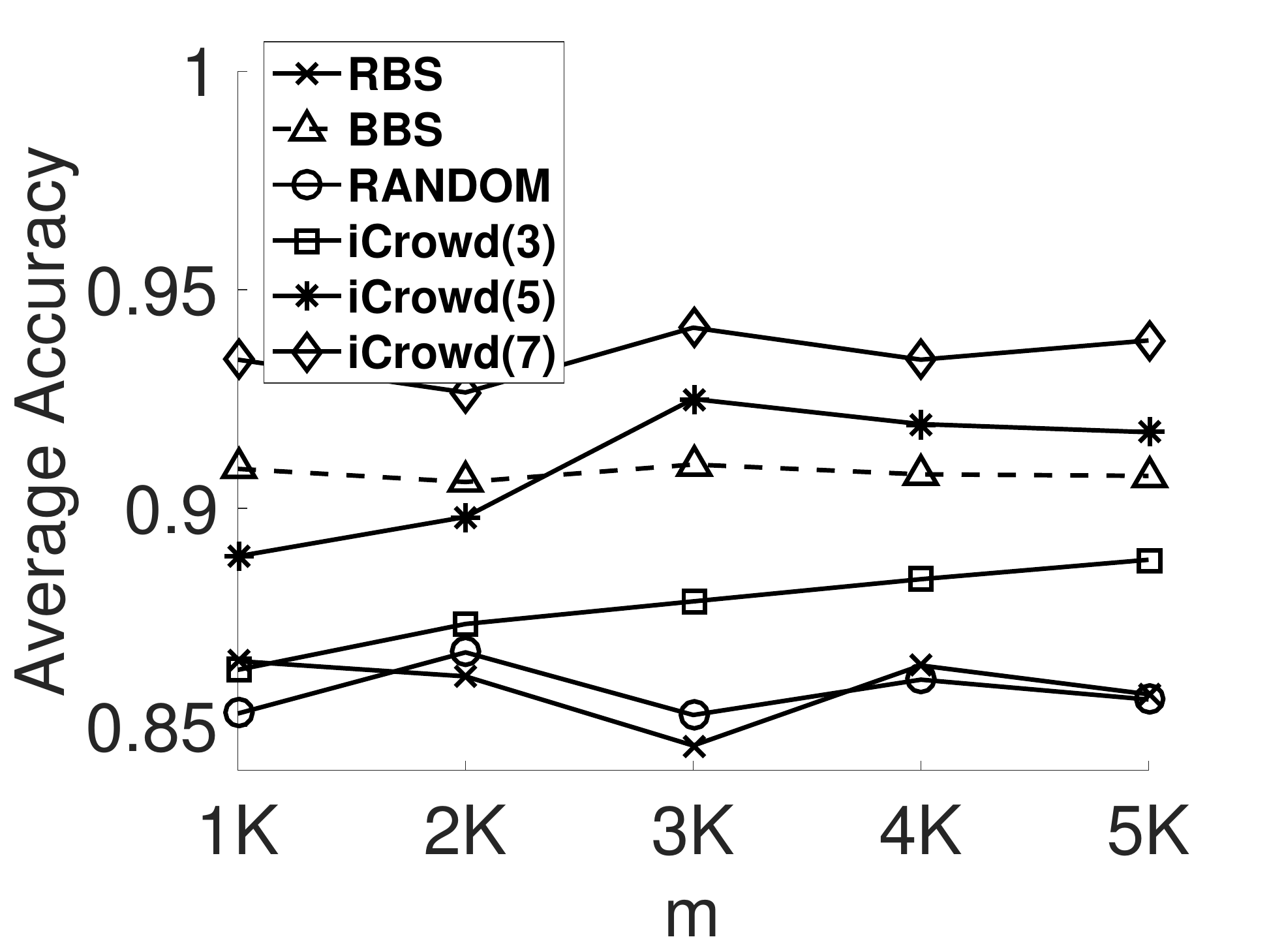}}
		\label{subfig:task_accuarcy_icrowd}}\vspace{-2ex}
	\caption{\small Effect of parameter $k$ of iCrowd on varying the number $m$ of tasks.}
	\label{fig:task_fig_icrowd}
\end{figure}
\noindent\textbf{Appendix C. The Experimental Results with iCrowd Using Different Parameter $k$ on Varying Number $m$ of Tasks.}
We show the results achieve by iCrowd algorithm using different parameter $k$ (the number of the workers assigned to each task) comparing with our task scheduler methods. We increase the value $k$ to 5 and 7. To better present the results, we show their results through varying the number of tasks.  In Figure \ref{fig:task_fig_icrowd}, the values of $k$ are recorded in the brackets behind ``iCrowd''. We present the maximum latencies of tested algorithms in Figure \ref{subfig:task_latency_icrowd}. We can see that when the parameter $k$ increases from 3 to 7, the maximum latencies of iCrowd increase obviously. Specifically, when $k=3$, the maximum latencies of iCrowd are lower than that of RANDOM but higher than that of our RBS and BBS. When $k>=5$, the maximum latencies of iCrowd are higher than that of RANDOM. As for the average accuracies shown in Figure \ref{subfig:task_accuarcy_icrowd}, when the parameter $k$ increases from 3 to 5, the average accuracies achieved by iCrowd also increase. The reason is that when the number $k$ of workers assigned to each task increases, the accuracies will increase according to Corollary 3.2. Specifically, when $k>=5$ and the number $m$ of tasks is larger than 2K, the average accuracies achieved by iCrowd are higher than that of our RBS and BBS. 

In conclude,  larger $k$ will lead to higher average accuracies achieved by iCrowd, however, and will also increase the maximum latencies of iCrowd.